\documentclass[reprint,superscriptaddress, pra]{revtex4-2}
\usepackage{amssymb,amsmath}
\usepackage{bm,bbm}
\usepackage{graphicx}
\usepackage[bookmarks=false]{hyperref}
\hypersetup{colorlinks=true,citecolor=blue,linkcolor=red,%
urlcolor=blue,pdfstartview=FitH,bookmarksopen=true}
\usepackage[T1]{fontenc}
\usepackage[osf,sc]{mathpazo}
\usepackage{dsfont}
\usepackage[ruled,vlined]{algorithm2e}
\usepackage{comment}
\usepackage{mathtools}
\usepackage{enumitem}
\usepackage{tikz}
\usetikzlibrary{positioning}
\usepackage{extarrows}

\newcommand{\ket}[1]{\mbox{$ | #1 \rangle $}}
\newcommand{\ketbra}[1]{\mbox{$ | #1 \rangle\langle #1 | $}}
\newcommand{\bra}[1]{\mbox{$ \langle #1 | $}}

\newcommand{\tr}{\mathrm{tr}}

\newcommand{\cP}{{\cal P}}

\newcommand{\E}{\mathrm{e}}

\definecolor{light-gray}{gray}{0.95}

\usepackage{amsthm}
\newtheoremstyle{note}      
  {\topsep/2}              	
  {\topsep/2}            	
  {}                        
  {\parindent}             	
  {\itshape}                
  {.---}                    
  {0pt}                     
  {\thmname{#1}\thmnumber{ \itshape#2}\thmnote{ (#3)}} 

\newtheorem{theorem}{Theorem}

\newtheorem{corollary}{Corollary}
\newtheorem{proposition}{Proposition}

\theoremstyle{definition}

\theoremstyle{remark}

\newtheoremstyle{remboldstyle}
  {}{}{\itshape}{}{\bfseries}{.}{.5em}{{\thmname{#1 }}{\thmnumber{#2}}{\thmnote{ (#3)}}}
\theoremstyle{remboldstyle}

\begin{document}
\title{Beating the Optimal Verification of Entangled States via Collective Strategies}

\author{Ye-Chao Liu}
\email{liu@zib.de}
\affiliation{Zuse-Institut Berlin, Takustra{\ss}e 7, 14195 Berlin, Germany}

\author{Jiangwei Shang}
\email{jiangwei.shang@bit.edu.cn}
\affiliation{Key Laboratory of Advanced Optoelectronic Quantum Architecture and Measurement (MOE), School of Physics, Beijing Institute of Technology, Beijing 100081, China}

\date{December 7, 2025}
%

\begin{abstract}
In the realm of quantum information processing, the efficient characterization of entangled states poses an overwhelming challenge, rendering the traditional methods including quantum tomography unfeasible and impractical.
To tackle this problem, we propose a new verification scheme using collective strategies, showcasing arbitrarily high efficiency that beats the optimal verification with global measurements. 
Our collective scheme can be implemented in various experimental platforms and scalable for large systems with a linear scaling on hardware requirement, and distributed operations are allowed. 
Notably, larger ensembles can always improve the efficiency further, but without increasing the quantum memory.
More importantly, the approach consumes only a few copies of the entangled states, while ensuring the preservation of unmeasured ones, and even boosting their fidelity for any subsequent tasks. 
Furthermore, our protocol provides additional insight into the specific types of noise affecting the system, thereby facilitating potential targeted improvements. 
These advancements hold promise for a wide range of applications, offering a pathway towards more robust and efficient quantum information processing.
\end{abstract}

\maketitle
%

\textit{Introduction.---}%
The rapid advancement of quantum technologies has highlighted the fundamental challenge of efficiently characterizing quantum systems.
In particular, entangled states are crucial for tasks including quantum teleportation \cite{Bennett.etal1993}, quantum key distribution \cite{Ekert1991, Bennett_QKD_2014}, and distributed or blind quantum computation \cite{Cirac.etal1999, Hayashi.Morimae2015, Gheorghiu_verification_2018}.
The standard method of quantum state tomography \cite{QSE2004} requires an exponential number of samples with respect to the system size, making it impractical for large systems. 
Although full tomography provides comprehensive knowledge about a quantum state, many practical tasks only require specific information, prompting the need to develop more efficient certification methods \cite{Flammia.Liu2011, Dimic.Dakic2018, Pallister.etal2018}.

The methodology of quantum state verification (QSV) stands out, in particular, owing to its high efficiency and low resource consumption \cite{Pallister.etal2018}. 
In short, QSV is a procedure to determine that the output of a quantum device is a specific target state, say $\ket{\psi}$. 
In general, the verification efficiency depends only on the second-largest eigenvalue $\lambda$ of the QSV scheme $\Omega$, such that 
\begin{equation}
    N \approx \frac1{1-\lambda}\epsilon^{-1}\ln\delta^{-1}\,,
\end{equation}
where $\epsilon$ is the infidelity and ${1-\delta}$ denotes the confidence level.
If one can perform global projective measurement ${\ket{\psi}\bra{\psi}}$, the target state can be easily verified with the sample complexity ${N_{\rm opt} \approx\epsilon^{-1}\ln\delta^{-1}}$, known as the optimal global verification efficiency. 
However, valuable target states are usually entangled, and such projections are experimentally challenging, especially for multipartite ones.
Nevertheless, a series of studies have, theoretically and experimentally, demonstrated that various entangled states can be verified by local measurements with polynomial, linear, or even constant scaling of the resource consumption relative to the system size \cite{Morimae.etal2017, Takeuchi.Morimae2018, Yu.etal2019, Li.etal2019, Wang.Hayashi2019, Zhu.Hayashi2019a, Zhu.Hayashi2019b, Zhu.Hayashi2019c, Zhu.Hayashi2019d, Liu.etal2019b, Li.etal2020b, Dangniam.etal2020, Zhang.etal2020a, Jiang.etal2020, Zhang.etal2020b, Li.etal2020a, Liu.etal2020b,Liu.etal2021, Han.etal2021, Zhu.etal2022, ChenHuang.etal2023, Li.Y.etal2021, Yu.etal2022}. 
Compared to tomography, QSV significantly reduces the resource consumption, and in many cases is comparable to global verification.  

Advanced measurement techniques can further enhance the verification efficiency. 
For instance, optimal verification protocols with adaptive measurements have been designed for arbitrary bipartite entangled states \cite{Yu.etal2019, Li.etal2019, Wang.Hayashi2019}. 
Universal optimal verification for arbitrary entangled states can be realized by employing quantum nondemolition measurements \cite{Liu.etal2020b}. 
Recently, collective measurements in QSV have shown the potential to surpass the optimal global verification for Bell states and Greenberger-Horne-Zeilinger (GHZ) states \cite{miguel-ramiro_collective_2022}. 
For general graph states, collective strategies can asymptotically reach the efficiency of global verification for the high-accuracy scenario as ${\epsilon\!\to\!0}$ \cite{chen_memoryQSV_2025}.
However, both these collective schemes are restricted to small scales with current technology, which limits their advantages.

Meanwhile, the price paid for QSV’s low resource consumption lies in its inability to learn noise, which is prevalent in various fast characterization methods. 
In practical scenarios, if the prepared quantum states fail the test, full tomography is often required to improve the state preparation, leading to high sample consumption again. 
This challenge has led to the development of improved tomographic methods including compressed sensing \cite{Gross.etal2010, Gross2011} and corrupted sensing \cite{ma_corruptedsensing_2025}, which, however, still scales exponentially though the resource consumption is reduced. 
Another approach to address imperfection involves understanding the noise channels, such as the Pauli channel, which typically requires exponentially increasing number of measurements to characterize. 
This learning efficiency can be improved to polynomial scaling with local correlations \cite{flammia_efficient_2020, harper_efficient_2020}, or to linear scaling under the assumption of sparse Pauli channels \cite{harper_fast_2021}.

In this work, we fully utilize the power of collective strategies by proposing a scheme of verifying arbitrary entangled states with arbitrarily high efficiency that beats the optimal global verification. 
The construction of the collective scheme is easy to realize experimentally and scalable with a linear increasing on hardware requirement.
Importantly, we can always scale the scheme with larger ensembles to improve the efficiency further, but without increasing the quantum memory, which is desirable for practical implementations. 
In the meantime, our approach can synchronously preserve and enhance the unmeasured states, thereby boosting their fidelity for any subsequent tasks, which is particularly desirable for online tasks.
Furthermore, the collective scheme can provide additional insight into the specific types of noise affecting the system, thereby facilitating targeted improvements.

\begin{figure}[t]
  \includegraphics[width=.95\columnwidth]{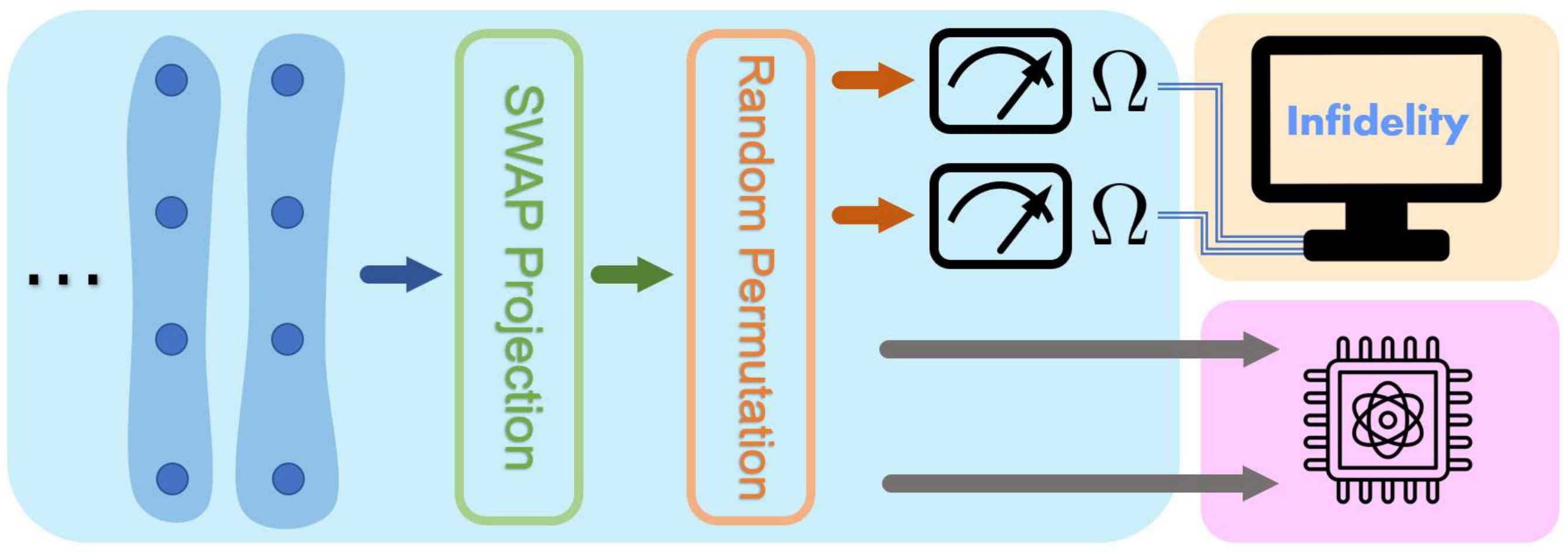}
  \caption{Sketch of the verification scheme using collective strategies. We assume that a source prepares multiple copies of the target state. Then a SWAP projection followed by a random permutation is applied on all the copies. Next, standard QSV is conducted on a random subset of the copies, while the remaining are retained for any subsequent tasks. If all the measured copies pass the test, the collective verification scheme is called successful. By repeating the procedure several rounds, the target state can be verified within a certain infidelity and confidence level.}
  \label{fig:strategy}
\end{figure}

\textit{The collective QSV scheme.---}%
An overview of the collective QSV scheme is depicted in Fig.~\ref{fig:strategy}. 
Our objective is to verify a device that claims to produce a target state $\ket{\psi}$ by collectively using $k$ copies of the states generated by the device. The procedure consists of the following three steps: 
(1) SWAP projection on all the $k$ copies; 
(2) random permutation; 
(3) standard QSV on a subset of $t(\leq k)$ copies.
Note that the last step is based on the tradeoff between time consumption and sample consumption of the verification procedure.

First, the SWAP projection on two quantum states is defined as ${D=(\openone+S)/2}$, where $S$ denotes the SWAP operation between the two states.
This procedure, known as the SWAP test, measures the similarity between two quantum states in terms of purity \cite{barenco_stabilisation_1996, buhrman_quantum_2001}. 
Thus, the SWAP projection is also useful for quantum state purification \cite{ricci_experimental_2004, Childs_2025_streaming}, which in turn can be beneficial for quantum cooling \cite{cotler_quantum_2019} and error mitigation \cite{Huggins_2021_virtual, Koczor_2021_exponential, czarnik_2021_qubit, Quek_2024_multivariate}.
For $k$ quantum states, the SWAP projection is extended to
\begin{eqnarray}
D_k=\frac1{2}\bigl(\openone+S_k\bigr)\,,
\end{eqnarray}
where $S_k$ represents the cyclic permutation on all the $k$ states by shifting the $m$th state to ${m\!-\!1\!\!\pmod{k}}$. 
Note that such a SWAP projection is a positive operator-valued measure (POVM) element.
One way to implement the SWAP projection on a state $\rho$ is to prepare and measure the ancilla of a controlled SWAP operation in the basis ${\ket{+}=(\ket{0}+\ket{1})/2}$ \cite{Ekert_estimation_2002}, i.e.,
\begin{eqnarray}
D_k(\rho) = \tr_{\rm ancilla}\left[\ket{+}\bra{+}\otimes\openone \cdot cS_k \bigl(\ket{+}\bra{+}\otimes \rho\bigr) cS_k^\dag\right]\!,
\end{eqnarray}
where $cS_k$ is the controlled permutation operation on the $k$ states collectively. While performing this operation on multiple multipartite states may seem challenging, it is feasible with current technologies \cite{Huggins_2021_virtual, Koczor_2021_exponential, czarnik_2021_qubit, Quek_2024_multivariate}; see the following proposition.
\begin{proposition}\label{Prop:SWAP}
The SWAP projection on $k$ (different or the same) $n$-qubit states can be realized using $nk$ controlled qubit-qubit SWAP operations (Fredkin gates), or $5nk$ two-qubit gates with one ancilla qubit.
\end{proposition}
In linear optics systems, the SWAP projection can be efficiently implemented using controlled qubit-qubit SWAPs, also known as the Fredkin gate \cite{patel_quantum_2016, ono_implementation_2017, starek_nondestructive_2018}. 
The Fredkin gate can be decomposed into five two-qubit gates \cite{yu_optimal_2015}, facilitating its efficient extension to a larger scale in general quantum circuit constructions. 
In cold-atom setups, a controlled SWAP on two multipartite states can be directly realized and extended to a permutation on $k$ states using Rydberg interactions \cite{pichler_measurement_2016}. 
See the sketch in Fig.~\ref{fig:SWAP}, and a detailed proof can be found in Appendix A \cite{supp}.

\begin{figure}[t]
  \includegraphics[width=.95\columnwidth]{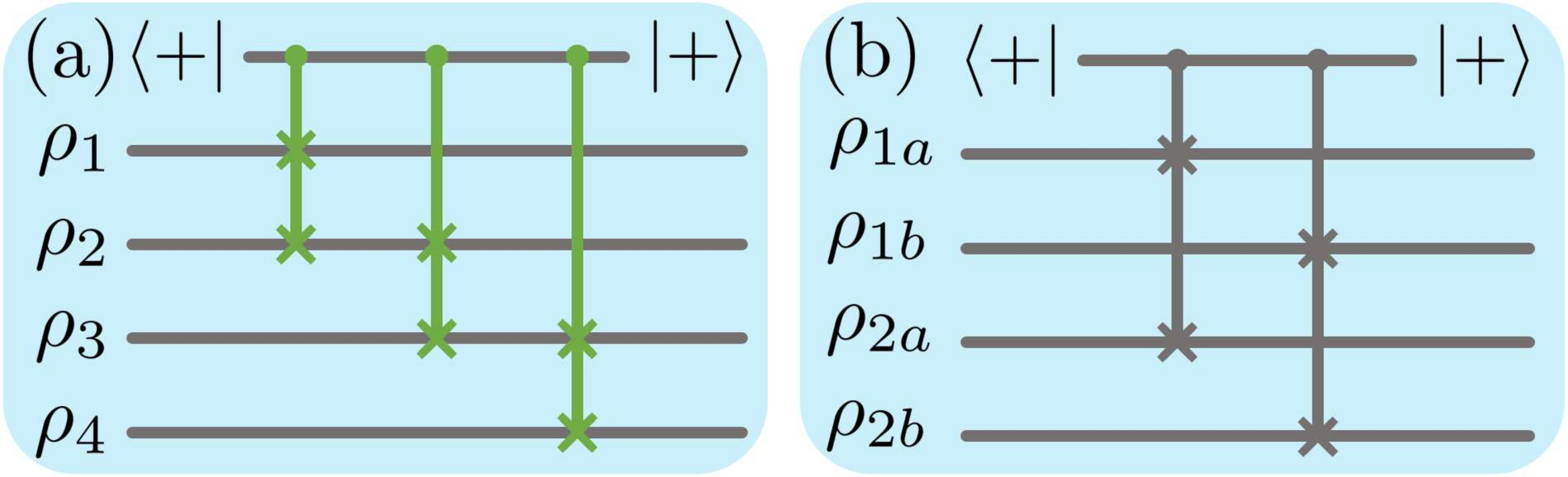}
  \caption{(a) Construction of a SWAP projection by performing controlled SWAP operations $cS_k$ on four input states. (b) Construction of a controlled SWAP operation by performing controlled qubit-qubit SWAPs, i.e., Fredkin gates on two bipartite states. }
  \label{fig:SWAP}
\end{figure}

If all the $k$ copies are indeed the target state $\ket{\psi}$, the SWAP projection will not disturb them
\begin{eqnarray}
    D_k\bigl(\ketbra{\psi}^{\otimes k}\bigr)
    = \ketbra{\psi}^{\otimes k}\,,
\end{eqnarray}
then ideally, selection of the measured copies should not influence the measurement outcomes.
Therefore, a random permutation $\cP$ can be applied before the measurements in order to mitigate potential asymmetrical noise.

Finally, using standard QSV measurement $\Omega$ on a subset of $t(\leq k)$ copies, the collective QSV scheme can be expressed as
\begin{eqnarray}
\Pi_{k,t}(\cdot) = \frac{1}{C_k^t}\sum_i\cP_i\left\{\bigl(\Omega^{\otimes t} \otimes \openone^{k-t}\bigr)\right\} D_k(\cdot)\,,
\end{eqnarray}
where $C_k^t$ is the binomial coefficient to normalize all the possible random permutations.

Considering that all the $k$ copies are either the target state $\ket{\psi}$ or a noisy state
\begin{eqnarray}\label{eq:white}
\sigma = (1-q)\ket{\psi}\bra{\psi} + q\openone/d\,,
\end{eqnarray}
the infidelity is given by ${\epsilon=q(d-1)/d}$ with ${d=2^n}$. 
This is a general form which any states can be transformed into through random operations \cite{foldager_can_2023, dalzell_random_2024}, usually satisfied in large-scale experiments \cite{urbanek_mitigating_2021, mi_information_2021}. For maximally entangled states, such a transformation can even retain the fidelity \cite{bennett_purification_1996, horodecki_reduction_1999}.
The following theorem presents the main result of our collective scheme.
\begin{theorem}\label{thm:main}
A target state $\ket{\psi}$ can be verified by the collective strategy $\Pi_{k,t}$ within infidelity $\epsilon$ and confidence level ${1-\delta}$ via
\begin{eqnarray}
M = \ln\delta^{-1}/\ln p^{-1} \approx \frac{2}{(1-\lambda)t+k}\epsilon^{-1}\ln\delta^{-1}
\end{eqnarray}
rounds of testing, and ${N = tM}$ number of samples.
Simultaneously, ${(k-t)M}$ copies of the output states $\sigma'$ with a smaller infidelity of $\epsilon/2$ are produced.
\end{theorem}
\begin{proof}
Here, we provide a brief proof for the simple case when ${t=1}$, and a detailed proof of the general case is postponed to Appendix B \cite{supp}. The white-noise mixed state in Eq.~\eqref{eq:white} can be rewritten as
\begin{eqnarray}
\sigma = (1-\epsilon)\ket{\psi}\bra{\psi} + \epsilon\bigl(\openone - \ket{\psi}\bra{\psi}\bigr)/(d-1)\,,
\end{eqnarray}
where ${\openone - \ket{\psi}\bra{\psi} = \sum_i \ketbra{\psi^\perp_i}}$ represents the orthogonal subspace of the target state $\ket{\psi}$. The standard QSV measurement $\Omega$ can always be decomposed into such an orthogonal subspace, thus we have
\begin{eqnarray}
\tr(\Omega\sigma) = (1-\epsilon) + \frac{1}{d-1}\sum_i\lambda_i\epsilon
\leq 1 - \epsilon + \lambda\epsilon\,,
\end{eqnarray}
where $\lambda_i$s are eigenvalues of $\Omega$, and ${\lambda = \max_i{\lambda_i}}$. 
The inequality is saturated for homogeneous measurements \cite{Zhu.Hayashi2019c, Zhu.Hayashi2019d, Liu_homo_2023}. 
When ${t=1}$, the only reduced state after the SWAP projection is $\sigma' = \bigl(\sigma + \sigma^k \bigr)/2$.

After random permutations, the passing probability for the white-noise mixed state of the collective scheme is
\begin{eqnarray}
p 
&=& 1/k\sum_i\tr\left[\cP_i\bigl\{(\Omega\otimes\openone)\bigr\} D_k\bigl(\sigma^{\otimes k}\bigr)\right] = \tr\bigl(\Omega\sigma'\bigr)\nonumber\\
&=& 1 - (k+1-\lambda)\epsilon/2 +O\bigl(\epsilon^2\bigr)\,.
\end{eqnarray}
Hence, if only one state in the ensemble is measured, the sample complexity is given by
\begin{eqnarray}
N = M = \ln\delta^{-1}/\ln p^{-1} \approx \frac{2}{1-\lambda+k}\epsilon^{-1}\ln\delta^{-1}\,.
\end{eqnarray}
And the fidelity of the unmeasured states can be obtained as ${F = \tr\!\left[\bigl(\Omega\otimes\ketbra{\psi}\otimes\openone\bigr) D_k\bigl(\sigma^{\otimes k}\bigr)\right]\!/p}$, so the infidelity is
\begin{eqnarray}
\epsilon' = \epsilon/2 + (k-1+\lambda)\epsilon^2/4 + O\bigl(\epsilon^3\bigr)\approx\epsilon/2\,.
\end{eqnarray}
\end{proof}

The passing probability of the collective scheme is comprised of two parts, namely the measurement on the ancilla of the SWAP projection and the measurement of standard QSV on the $t$ samples. 
The verification is considered as successful only when both of the measurements succeed. 
The power of our collective strategy can be showcased even in the simplest setting $\Pi_{2,1}$, where the efficiency $2/(1-\lambda+2)$ is strictly smaller than that of the optimal global verification.

In principle, implementing the collective strategy on a larger ensemble can improve the efficiency infinitely in terms of both the time and sample consumption, i.e., 
\begin{eqnarray}
    \lim_{k \to \infty} M = 0\,, \qquad \lim_{k \to \infty,\, k \gg t} N = 0\,.
\end{eqnarray}
A larger measured subset can speed up the verification but worsen the sample complexity. 
While, the improvement of unmeasured states, without strong dependence on the size of the ensemble and measured subset, can always be maintained. 
Notably, as long as no more than half of the ensemble is measured, the collective scheme can beat the optimal global verification in both the time and sample consumption, i.e.,
\begin{eqnarray}
    M \leq N < N_{\rm opt}\,,\qquad \text{for~} 1 \leq t \leq k/2\,.
\end{eqnarray}
Even counting in both the measured and unmeasured states, the sample complexity of the collective scheme remains manageable. 
For instance, the most efficient scenario is to take ${t=1}$, namely $\Pi_{k,1}$, in total
\begin{eqnarray}
    kM \approx \frac{2}{(1-\lambda)/{k}+1}\epsilon^{-1}\ln\delta^{-1} < 2\epsilon^{-1}\ln\delta^{-1}\,
\end{eqnarray}
samples are needed.
This property leads to various benefits for practical applications, of which we discuss in the following.

One notes that the SWAP projection $D_k$ extends the operation in Ref.~\cite{ricci_experimental_2004}, which purifies single qubits. 
Our generalization is designed to act only on the two nearest-neighbor quantum states, making it experimentally more feasible, as illustrated in Fig.~\ref{fig:SWAP}. 
Beyond the challenge of implementing collective operations, quantum memory is another significant limit in practice. 
In principle, quantum memory is only needed to temporarily store quantum states prior to measurements. 
Leveraging this, our collective scheme can be implemented by sequentially interacting one registered quantum state with others, measuring only the registered one, then discarding the others. 
This variation is equivalent to the one constructed with $D_k$, which can be scaled up to achieve arbitrarily high efficiency, but without increasing the memory.
See Appendix G for more details \cite{supp}.

\textit{The general scenario.---}%
Up to this point, we have considered the ensemble with independent distribution under the white noise model. The generality of this consideration with other types of noise is detailed in Appendix C \cite{supp}. 
In this section, we move on to discuss the scenario where correlated noise exists within our collective strategy. 
In standard QSV frameworks, the correlation between multiple samples is treated as an adversarial scenario, such that a powerful adversary can control all the samples throughout the entire verification procedure which leads to more sample consumption \cite{Zhu.Hayashi2019c, Zhu.Hayashi2019d}.

Here, we assume a more realistic scenario that the correlated noise (or an adversary) can affect at most the entire ensemble of each round independently, that is, $k$ samples. 
We model the global white noise on the ensemble as
\begin{eqnarray}
    \eta = (1-q)\ket{\psi}\bra{\psi}^{\otimes k} + q\openone/d^k\,,
\end{eqnarray}
where the infidelity of each sample is ${\epsilon=q(d-1)/d}$. 
Note that if one focuses on the individual sample in each ensemble, it reduces to the scenario under independent white noise as in Eq.~\eqref{eq:white}, making the correlation indistinguishable without the collective strategy.
For the correlated noise, we have the following theorem.
\begin{theorem}\label{thm:Gwhite}
A target state $\ket{\psi}$ can be verified by the collective strategy $\Pi_{k,t}$ under global white noise within infidelity $\epsilon$ and confidence level ${1-\delta}$ via
\begin{eqnarray}
    M \approx \frac{2}{(1-\lambda)t+1}\epsilon^{-1}\ln\delta^{-1}
\end{eqnarray}
rounds of testing, and ${N=tM}$ number of samples.
Simultaneously, ${(k-t)M}$ copies of the output states $\sigma'$ with a smaller infidelity of $\epsilon/2$ are produced.
\end{theorem}
\begin{proof}
The proof is similar to that of Theorem~\ref{thm:main}; see Appendix D for details \cite{supp}.
\end{proof}
Under global white noise, the collective scheme does not surpass the optimal global verification in terms of sample complexity, but remains comparable by a constant factor of $2$ with the one-sample scheme $\Pi_{k,1}$, i.e.,
\begin{eqnarray}
    N < 2 N_{\rm opt}\qquad\text{for~}  t=1\,,\forall k\geq 2\,.
\end{eqnarray}
Selecting a larger measured subset can improve the verification speed infinitely, but also increases the sample consumption. 
Meanwhile, the infidelity of unmeasured states can exponentially decrease with the subset size $t$ as $a^t \epsilon/2$, where ${a<1}$ depends on the target state; see Appendix D for details \cite{supp}. 

In scenarios where an adversary can control the entire quantum system with arbitrary global operations, verification becomes more challenging, thus more sample consumption is needed \cite{Zhu.Hayashi2019c}.
However, as long as the adversary cannot control everything, but limited to the entire ensemble (even for large $k$), the collective scheme can guarantee the unmeasured output state within infidelity $\epsilon$ and confidence level ${1-\delta}$ by consuming ${1/(1-\lambda)\epsilon^{-1}\ln\delta^{-1}}$ samples, no more than the requirement of standard QSV under no adversary.
See Appendix E for further discussions \cite{supp}.

\begin{figure}[t]
  \includegraphics[width=0.95\columnwidth]{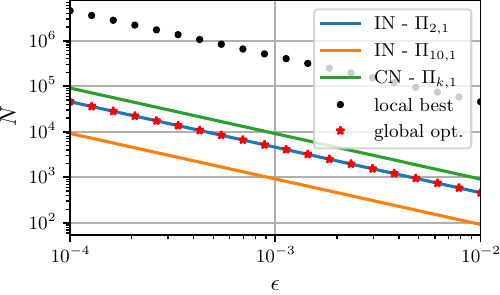}
  \caption{Sample complexity for verifying a $100$-qubit Dicke state using the collective strategies under independent noise with schemes $\Pi_{2,1}$ (blue) and $\Pi_{10,1}$ (orange), and under correlated noise with the scheme $\Pi_{k,1}$ (green), as compared to the best known local scheme \cite{Li.etal2020a} (black dot) and the optimal global scheme (red star). The confidence level is set to $1-\delta=99\%$.}
  \label{fig:p}
\end{figure}

\textit{Applications.---}%
In practical implementations of standard QSV, the process begins by randomly selecting a portion of the prepared samples for verification. 
Random selection is crucial to prevent potential adversarial manipulations, where carefully crafted samples could pass the verification while others could deceive the users. 
However, our collective scheme allows for a synchronous verification and utilization of the target states, which is particularly advantageous for online tasks. 
By beating the optimal global verification, our approach ensures that the sample complexity remains manageable for any target state. 
For instance, in order to guarantee ${\epsilon=1\%}$ and ${1-\delta=99\%}$, no more than ${\lceil \epsilon^{-1}\ln\delta^{-1} \rceil = 461}$ rounds of the verification are required by using the scheme $\Pi_{k,t}$ ($t\leq k/2$), and larger ensembles lead to less time. 
If a task demands $1024$ copies of a verified target state, the simplest setup $\Pi_{2,1}$ needs to consume another $1024$ copies and $1024$ rounds of tests.
With larger ensembles, for instance the scheme $\Pi_{9,1}$, only $128$ additional samples and $128$ rounds of tests are needed.

Depending on the specific target state, different levels of improvement can be achieved by the collective scheme.
For the simplest case of Bell states, the optimal QSV scheme with local operations has a sample complexity of ${N_{\rm Bell}\approx \frac{3}{2}\epsilon^{-1}\ln\delta^{-1}}$.
Using $2k$ controlled qubit-qubit SWAP operations, our collective scheme $\Pi_{k,1}$ can beat the optimal global verification with the sample complexity ${{N\approx\frac{3}{4}\epsilon^{-1}\ln\delta^{-1} < N_{\rm opt}}<N_{\rm Bell}}$. 

Nevertheless, a significant advantage of our collective scheme lies in the verification of more complex multipartite entangled states, whose projection is difficult to realize (if not simply impossible) and the sample consumption increases rapidly with the system size.
In Fig.~\ref{fig:p}, we consider the verification of a $100$-qubit Dicke state.
Our collective scheme with the simplest setup $\Pi_{2,1}$ beats the optimal global verification,  and outperforms the best local QSV scheme by two orders of magnitude \cite{Li.etal2020a}. 
A larger ensemble with the scheme $\Pi_{10,1}$ provides an additional five-fold improvement. 
Moreover, even under the global noise on the entire ensemble, which is typically difficult to verify with standard QSV, our approach still demonstrates a notable improvement.
Details including the enhancement of unmeasured states and the tradeoff between the time and sample consumption are discussed in Appendix F \cite{supp}.

\textit{Discussions.---}%
The collective strategies in QSV not only conserve resources but also offer the additional advantage of distinguishing different noise types. 
Simply put, standard QSV performs a certain number of successful rounds of testing $M_s$ to verify the entangled state within infidelity $\epsilon$ and confidence level $1-\delta$.
However, due to the higher efficiency of the collective scheme, states with independent white noise may not pass with the same number of rounds $M_s$, enabling us to differentiate correlated noise from independent noise. 
With additional samples and data processing, our collective can thus identify noise variations, aiding targeted state preparation.
Be noted that this method is general and applicable to various noise types; see Appendix F for examples \cite{supp}.

In real experiments \cite{ Zhang.etal2020a,Zhang.etal2020b, Jiang.etal2020}, noise is inevitable, nevertheless an entangled state can still be verified within infidelity $\epsilon$ and significance level ${\exp(-D[f||(1-\epsilon+\lambda\epsilon)]N_{\rm total})}$ based on the overall passing probability $f$ out of $N_{\rm total}$ tests, where $D(x||y)=x\log\frac{x}{y}+(1-x)\log\frac{1-x}{1-y}$ is the Kullback–Leibler divergence \cite{Yu.etal2019}. 
Importantly, in QSV, the infidelity $\epsilon$ incorporates all errors, including gate noise and imperfect state preparation. Consequently, noisy gates do not degrade the verification efficiency, rather, prior knowledge of gate noise can reduce the sample consumption, preserving the advantage of our collective scheme even in noisy settings, as demonstrated with a specific example in Appendix G \cite{supp}.

In principle, our collective scheme is fundamentally effective for verifying entangled states with infidelity ${\epsilon\!<\!1/2}$. This limitation corresponds to a threshold for state purification, which is a reasonable assumption in practical applications.
Apart from the arbitrarily high verification efficiency with larger ensembles as ${k\!\to\!\infty}$, another advantage of our scheme is the straightforward extension to arbitrary multipartite entangled states. This contrasts with the work of Ref.~\cite{miguel-ramiro_collective_2022}, which demands entanglement purification and works only for specific types of entangled states.

One notices that a notable feature of our scheme is the use of a single ancillary qubit only. While this approach conserves resources, it poses challenges for constructing distributed verification protocols for multipartite entangled states. However, this issue can be addressed by employing a maximally entangled state on qubits instead of a single ancillary qubit. 
With the entangled ancilla, which does not need to be perfect, our distributed collective scheme is similar to that of Refs.~\cite{miguel-ramiro_collective_2022, chen_memoryQSV_2025}, which can be realized locally, yet no higher-dimensional local systems are needed. A detailed comparison with related works can be found in Appendix G \cite{supp}.

\textit{Summary.---}%
We have proposed an arbitrarily high efficient scheme for quantum state verification using collective strategies. 
Drawing inspiration from quantum state purification, our collective QSV scheme not only outperforms the optimal global verification for arbitrary entangled states but also enhances the unmeasured states for any subsequent tasks. 
Importantly, the arbitrarily high efficiency can be achieved by scaling our scheme with larger ensembles, but without increasing the quantum memory. Our scheme is also suitable for larger systems with a linear scaling on hardware requirement.
Additionally, our collective scheme is capable of distinguishing different noise types, a unique feature that non-collective strategies cannot have. 
This capability is advantageous for targeted improvements for state preparation. 
Furthermore, we demonstrated the applicability of online tasks and the significant efficiency improvement via specific examples, particularly for large-scale complex multipartite entangled states.

Numerous intriguing aspects of the collective scheme remain unexplored and warrant further investigation. 
For instance, the collective scheme can be naturally generalized to higher-dimensional systems, which needs specific discussion on generalized SWAP projections. 
Also, the SWAP projection can be enhanced through optimization of the entangling sequence, to further improve the performance of the collective scheme. 
There is also a one-to-one correspondence between different purification protocols and collective strategies, which might yield additional benefits. 
Furthermore, collective operations can be effectively combined with other advanced yet well-established measurement techniques, such as nondemolition measurements, to achieve more superior performance.

\acknowledgments
We are grateful to L. Vandr\'{e} for interesting discussions.
This work was supported by the National Natural Science Foundation of China (Grants No.~92265115 and No.~12175014) and the National Key R\&D Program of China (Grant No.~2022YFA1404900).
Y.-C. Liu is also supported by the DFG Cluster of Excellence MATH+ (EXC-2046/1, Project No.~390685689) funded by the Deutsche Forschungsgemeinschaft (DFG).

%

\onecolumngrid
\newpage

\appendix
In this Supplemental Material, we provide further details about our collective scheme, which beats the optimal verification for arbitrary entangled states.
In Appendix~\ref{App:SWAP}, we explain how to realize the SWAP projection in details, focusing on the experimental feasibility and scalability of the construction. 
Appendix~\ref{App:ProofThm1} presents the complete proof of the main result, i.e., Theorem~\ref{thm:main}. In the following Appendix~\ref{App:GeneralWhite}, a discussion on the generality of our scheme by considering independent and identically distributed white noise is supplemented. 
Other important cases including the global white noise and global unitary control are discussed in Appendix~\ref{App:ProofThm2} and Appendix~\ref{App:GU}, respectively. 
In Appendix~\ref{App:DetailsApp}, more results are provided by applying our collective scheme on various entangled states.
Finally in Appendix~\ref{App:Comp}, by comparing with other related works, we highlight the advantages of our collective scheme in several aspects.

\section{Realization of the SWAP projection}\label{App:SWAP}
Here we show how to realize a general SWAP projection on $k$ quantum states
\begin{eqnarray}
    D_k=\bigl( \openone+S_k \bigr)/2\,,
\end{eqnarray}
where $S_k$ is the cyclic permutation on all the $k$ states such that the $m$th state is shifted to ${m\!-\!1\!\!\pmod{k}}$. 

The SWAP projection is equivalent to preparing and measuring the ancilla of a controlled SWAP operation under the basis ${\ket{+}=(\ket{0}+\ket{1})/2}$, i.e.,
\begin{eqnarray}
    D_k(\rho) = \tr_{\rm ancilla}\left[\ket{+}\bra{+}\otimes\openone \cdot cS_k \bigl(\ket{+}\bra{+}\otimes \rho \bigr) cS_k^\dag\right]\!,
\end{eqnarray}
where $cS_k$ is the controlled permutation (SWAP) operation on the $k$ states collectively. 
In cold-atom setups \cite{pichler_measurement_2016}, a controlled SWAP operation on two multipartite states
\begin{eqnarray}
    cS_2=cS_{{\rm ancilla},A,B}=\ket{0}\bra{0}_{\rm ancilla}\otimes \openone_{AB}+\ket{1}\bra{1}_{\rm ancilla}\otimes S_{AB}
\end{eqnarray}
can be realized directly, thus we can extend the permutation as
\begin{eqnarray}
    cS_k = \ket{0}\bra{0}_{\rm ancilla}\otimes \openone_{1,2,\cdots,k}+\ket{1}\bra{1}_{\rm ancilla}\otimes S_{1,2,\cdots,k} = cS_{{\rm ancilla},k,k-1} \otimes \cdots \otimes cS_{{\rm ancilla},2,1}\,.
\end{eqnarray}
As shown in Fig.~\ref{fig:SWAP}(a), with the preparation and measurement on the ancilla under the basis $\ket{+}$, we can realize the SWAP projection in such cold-atom platforms.
\begin{corollary}\label{coro:1}
     The SWAP projection on $k$ multipartite quantum states can be realized with $k$ controlled SWAP operations with one ancilla qubit. 
\end{corollary}

In addition, if there exists the constraint that only controlled qubit-qubit SWAP operations are allowed, known also as the Fredkin gate in linear optics systems \cite{patel_quantum_2016, ono_implementation_2017, starek_nondestructive_2018}, we can construct the controlled SWAP operation on two $n$-qubit states as
\begin{eqnarray}
    cS_2 = \bigotimes_{t=1}^n cS_2^{(t)}
    =\ket{0}\bra{0}_{\rm ancilla}\otimes \openone^{(1)} \otimes \cdots \otimes \openone^{(n)}+\ket{1}\bra{1}_{\rm ancilla}\otimes S_2^{(1)} \otimes \cdots \otimes S_2^{(n)}\,,
\end{eqnarray}
where $cS_2^{(t)}$ is the controlled qubit-qubit SWAP operation on the $t$th qubit; see Fig.~\ref{fig:SWAP}(b). 
However, these linear optics gates are probabilistic and it is not possible to easily concatenate them as they usually operate in coincidence bases.

One alternative method to realize the Fredkin gate is to use the Toffoli gate with an additional controlled-$X$ gate.
Or we can directly decompose the Fredkin gate into five two-qubit gates \cite{yu_optimal_2015}, then together with Corollary \ref{coro:1}, we recover Proposition~\ref{Prop:SWAP} in the main text. Similar results have also been discussed in the field of error mitigation \cite{Huggins_2021_virtual, Koczor_2021_exponential, czarnik_2021_qubit, Quek_2024_multivariate}.

Finally, one notes that, for ${k=2}$ and ${n=1}$, the SWAP projection is equivalent to the operation in Ref.~\cite{ricci_experimental_2004}, which realizes the purification of single qubits.
Though it is not directly extendable for entangled states as 
\begin{eqnarray}
    \frac{\openone+S_A}{2}\otimes\frac{\openone+S_B}{2}\neq\frac{\openone+S_{AB}}{2}\,,
\end{eqnarray}
with more complex coincidence bases, projection onto symmetric subspace of entangled states might be possible. 
This in fact offers another potential construction method for our collective scheme, of which no ancilla is demanded.

\section{Proof of Theorem~\ref{thm:main}}\label{App:ProofThm1}
\begin{proof}
    To verify a certain target state $\ket{\psi}$, the construction of standard QSV $\Omega$ demands that $\ket{\psi}$ can always pass the tests. Hence, $\Omega$ can be written as 
    \begin{eqnarray}
        \Omega = \ketbra{\psi} + \sum_i \lambda_i \ketbra{\psi_i^\perp}\,,
    \end{eqnarray}
    where $\ket{\psi_i^\perp}$s are the eigenstates in the orthogonal subspace $\openone\!-\!\ket{\psi}\bra{\psi}$, and $\lambda_i$s are the corresponding eigenvalues that $\lambda_i < 1, \forall i$. 
    The white noisy state is equivalent to
    \begin{eqnarray}
        \sigma = (1-\epsilon)\ket{\psi}\bra{\psi} + \epsilon \bigl(\openone - \ket{\psi}\bra{\psi}\bigr)/(d-1)\,,
    \end{eqnarray}
    thus the passing probability of measuring one sample by the standard QSV is given by
    \begin{eqnarray}
        \tr(\Omega\sigma) = (1-\epsilon) + \frac{1}{d-1}\sum_i\lambda_i\epsilon
        \geq 1 - \epsilon + \lambda\epsilon\,,
    \end{eqnarray}
    where ${\lambda=\max_i{\lambda_i}}$. The inequality is saturated for the homogeneous QSV protocols \cite{Zhu.Hayashi2019c, Zhu.Hayashi2019d, Liu_homo_2023}. 
    After the SWAP projection ${D_k=\bigl(\openone+S_k\bigr)/2}$, where $S_k$ permutes $m$th state to ${m\!+\!1\!\!\pmod{k}}$, the $k$ quantum states are
    \begin{eqnarray}
        D_k \bigl(\sigma_1\otimes\sigma_2\otimes\cdots\sigma_k \bigr)
        &=& \frac{1}{4}\sigma_1\otimes\sigma_2\otimes\cdots\sigma_k 
        + \frac{1}{4}\sigma_2\otimes\sigma_3\otimes\cdots\sigma_1 \nonumber\\
        &+& \frac{1}{4}S_k\cdot\sigma_1\otimes\sigma_2\otimes\cdots\otimes\sigma_k
        + \frac{1}{4}\sigma_1\otimes\sigma_2\otimes\cdots\otimes\sigma_k\cdot S_k^\dagger\,.
    \end{eqnarray}
    \begin{figure}[t]
    \includegraphics[width=0.55\columnwidth]{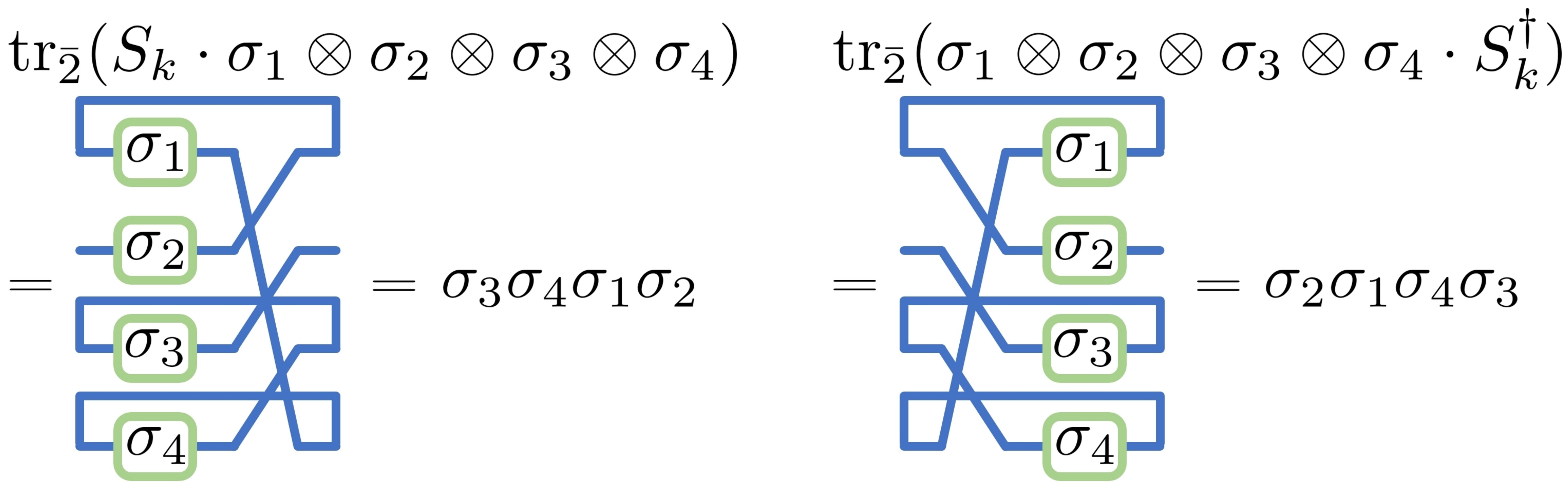}
    \caption{Tensor diagram for solving the (partial) trace of $\bigl(S_k \cdot \sigma_1 \otimes\cdots\otimes \sigma_k \bigr)$. The permutation $S_k$ and its conjugate $S_k^\dagger$ are represented by the cyclic permutations of the indices (legs) of the $\sigma$s.}
    \label{fig:tr}
    \end{figure}
    With the graphical illustration in Fig.~\ref{fig:tr}, when $t=1$, the $m$th reduced state of the $k$-sample ensemble is 
    \begin{eqnarray}\label{eq:Dk}
        \sigma'_m 
        &=& \tr_{\bar{m}}\bigl[D_k\bigl(\sigma_1\otimes\sigma_2\otimes\cdots\sigma_k\bigr)\bigr]\nonumber\\
        &=& \frac{1}{4}\sigma_m 
        + \frac{1}{4}\sigma_{(m+1|k)} 
        + \frac{1}{4}\tr_{\bar{m}}\bigr(S_k\cdot\sigma_1\otimes\sigma_2\otimes\cdots\otimes\sigma_k \bigl) 
        + \frac{1}{4}\tr_{\bar{m}} \bigl(\sigma_1\otimes\sigma_2\otimes\cdots\otimes\sigma_k\cdot S_k^\dagger\bigr) \nonumber\\
        &=& \frac{1}{4}\sigma_m 
        + \frac{1}{4}\sigma_{(m+1|k)}
        + \frac{1}{4}\sigma_{(m+1|k)}\sigma_{(m+2|k)}\cdots\sigma_{(m|k)}   
        + \frac{1}{4}\sigma_{(m|k)}\sigma_{(m-1|k)}\cdots\sigma_{(m-k+1|k)}\,,
    \end{eqnarray}
    where the index ${(i|k) \equiv i\!\!\pmod{k}}$. As all the samples are promised to be the same, the reduced state is
    \begin{eqnarray}
        \sigma' = \bigl(\sigma + \sigma^k\bigr)/2\,,\quad{\rm where~}\sigma^k := \sigma \cdots \sigma \,\,\text{($k$ times)}\,.
    \end{eqnarray}
    Hence when $t=1$, the passing probability can be directly calculated as 
    \begin{eqnarray}
        p 
        &=& 1/k\sum_i\tr\left[\cP_i\left\{\bigl(\Omega\otimes\openone\bigr)\right\} D_k\bigl(\sigma^{\otimes k}\bigr)\right] = \tr\bigl(\Omega\sigma'\bigr)\nonumber\\
        &=& \left[1 - \epsilon + \lambda\epsilon + (1-\epsilon)^k + \frac{\epsilon^k\lambda}{(d-1)^{k-1}}\right]/2 \nonumber\\
        &=& 1 - (k+1-\lambda)\epsilon/2 +O(\epsilon^2)\,.
    \end{eqnarray}
    In general, with an arbitrary selection of the $t$-sized subset to measure, we have
    \begin{eqnarray}
        \tr\left[\mathcal{P}_i\left\{\bigl(\Omega^{\otimes t}\otimes\openone\bigr)\right\}D_k\bigl(\sigma^{\otimes k}\bigr)\right] 
        = \tr\left[\bigl(\Omega^{\otimes t}\otimes\openone\bigr)D_k\bigl(\sigma^{\otimes k}\bigr)\right]
        = \frac{1}{2}\Bigl[\tr(\Omega\sigma)\Bigr]^t 
        + \frac{1}{2}\tr\left[(\Omega\sigma)^t\sigma^{k-t}\right]\!.
    \end{eqnarray}
    Note that we have used the commutability between $\Omega\sigma$ and $\sigma$ here.
    Thus, the passing probability of measuring $t$ samples in the $k$-sample ensemble by the collective QSV is 
    \begin{eqnarray}\label{eq:p_iid}
        p = 1/C_k^t\sum_i\tr\left[\cP_i\left\{\bigl(\Omega^{\otimes t}\otimes\openone\bigr)\right\}D_k\bigl(\sigma^{\otimes k}\bigr)\right] = \frac{1}{2}\left[\tr(\Omega\sigma)\right]^t 
        + \frac{1}{2}\tr\left[(\Omega\sigma)^t\sigma^{k-t}\right]\!.           
    \end{eqnarray}
    And for verifying an entangled state under independent white noise, we have
    \begin{eqnarray}\label{eq:p_IN}
        p_{\rm IN} = \left[(1 - \epsilon + \lambda\epsilon)^t + (1-\epsilon)^k + \frac{\epsilon^k\lambda^t}{(d-1)^{k-1}}\right]/2
        = 1 - [k+t(1-\lambda)]\epsilon/2 +O(\epsilon^2)\,.
    \end{eqnarray}
    Therefore, the time complexity and sample complexity are given by
    \begin{eqnarray}
        M &=& \ln\delta^{-1}/\ln p_{\rm IN}^{-1} \approx \frac{2}{(1-\lambda)t+k}\epsilon^{-1}\ln\delta^{-1}\,,\\
        N &=& tM \approx \frac{2t}{(1-\lambda)t+k}\epsilon^{-1}\ln\delta^{-1}\,.
    \end{eqnarray}
    The fidelity of one of the unmeasured states is 
    \begin{eqnarray}
        F &=& \tr\left[\bigl(\Omega^t\otimes\ketbra{\psi}\otimes\openone\bigr) D_k\bigl(\sigma^{\otimes k}\bigr)\right]/p_{\rm IN} \nonumber\\
        &=& \left\{\frac{1}{2}\Bigl[\tr(\Omega\sigma)\Bigr]^t \bra{\psi}\sigma\ket{\psi}
        + \frac{1}{2}\tr\left[(\Omega\sigma)^t\ketbra{\psi}\sigma^{k-t}\right]\right\}/p_{\rm IN}\nonumber\\
        &=& \left[(1-\epsilon)(1-\epsilon+\lambda\epsilon)^t+(1-\epsilon)^k\right]/(2p_{\rm IN})\,,
    \end{eqnarray}
    and the infidelity is reduced to
    \begin{eqnarray}
        \epsilon' \approx \epsilon \frac{(1-\epsilon+\lambda\epsilon)^t}{(1-\epsilon+\lambda\epsilon)^t+(1-\epsilon)^k}= \frac{1}{2}\epsilon + \frac{1}{4}[k - (1-\lambda)t]\epsilon^2 + O(\epsilon^3)\,.
    \end{eqnarray}
    Note also that the following limits 
    \begin{eqnarray}
        \lim_{k\to\infty}\epsilon' &=& \epsilon\,,\\
        \lim_{t=k\to\infty}\epsilon' &=& \epsilon\,,
    \end{eqnarray}
    hold asymptotically. 
\end{proof}

\section{The generality of white noise}\label{App:GeneralWhite}
White noise is also known as the noise under global depolarizing channel, and for maximally entangled states, it is equivalent to the local depolarizing channel. 
In general, many other noise models correspond to different practical scenarios. 
For instance, an alternative noise model
\begin{eqnarray}
\sigma_{\rm mix} = (1-\epsilon)\ketbra{\psi}+\epsilon\ketbra{\psi^\perp}
\end{eqnarray}
fits many types of noise, including for instance, the amplitude and phase damping channels on maximally entangled states. 
With this noise model, the passing probability of the collective scheme is given by
\begin{eqnarray}
p_{\rm mix} &=& \left[(1 - \epsilon + \lambda\epsilon)^t + (1-\epsilon)^k + \epsilon^k\lambda^t\right]/2\,,
\end{eqnarray}
which is slightly different from the white noise case in the order of $O(\epsilon^k)$. Thus, the main results of our analysis cover a wide range of practical scenarios.

In addition to the environmental noise, imperfect control of operations can lead to
\begin{eqnarray}
\sigma_\text{\sc ur}  = \ketbra{\psi_\epsilon}\,,\quad
\ket{\psi_\epsilon} = \sqrt{1-\epsilon}\ket{\psi}+\sqrt{\epsilon}\ket{\psi^\perp}\,,
\end{eqnarray}
which represents a unitary rotation of the target state. This is also considered as the worst-case scenario in standard QSV \cite{Pallister.etal2018}, with the passing probability of the collective strategy given by
\begin{eqnarray}
p_\text{\sc ur} = (1 - \epsilon + \lambda\epsilon)^t\,,
\end{eqnarray}
which leads to a lower efficiency for both the time and sample complexity.
However, when the target states are maximally entangled, we can always apply the isotropic twirling operation (local random operations) to transform such a noisy state into the white noise type without changing the fidelity \cite{bennett_purification_1996, horodecki_reduction_1999}. 
For other target states, by adding random circuits, any local noise can be transformed into global white noise with a worse fidelity \cite{foldager_can_2023, dalzell_random_2024}. 
Especially true and applicable is for large complex quantum systems, due to the preparation procedure, global white noise can be commonly assumed \cite{urbanek_mitigating_2021, mi_information_2021}.

\section{Proof of Theorem~\ref{thm:Gwhite}}\label{App:ProofThm2}
\begin{proof}
    The SWAP projection satisfies the linearity such that, for the globally correlated white noisy state, 
    \begin{eqnarray}
        D_k(\eta) = (1-q)D_k\bigl(\ketbra{\psi}^{\otimes k}\bigr) + q D_k\bigl(\openone/d^k\bigr)\,,\qquad q=\frac{d\epsilon}{d-1}\,.
    \end{eqnarray}
    Then using Eq.~\eqref{eq:p_iid} in the proof of Theorem \ref{thm:main}, we have 
    \begin{eqnarray}
        \tr\left[\bigl(\Omega^{\otimes t}\otimes\openone\bigr) D_k\bigl(\ketbra{\psi}^{\otimes k}\bigr)\right] &=& 1\,,\\
        \tr\left[\bigl(\Omega^{\otimes t}\otimes\openone\bigr) D_k\bigl(\openone/d^k\bigr)\right] &=& \frac{1}{2}\left[\tr(\Omega\frac{\openone}{d})\right]^t 
        + \frac{1}{2}\tr\left[(\Omega\frac{\openone}{d})^t(\frac{\openone}{d})^{k-t}\right]\nonumber\\
        &=&\left(\left[\lambda+\frac{1-\lambda}{d}\right]^t + \frac{(d-1)\lambda^t+1}{d^k}\right)/2\,.
    \end{eqnarray}
    A direct derivation gives us the passing probability of the globally correlated noisy state as
    \begin{eqnarray}\label{eq:p_CN}
        p_{\rm CN} 
        &=& 1-\frac{d\epsilon}{d-1} 
        + \frac{d\epsilon}{2(d-1)}\left(\left[\lambda+\frac{1-\lambda}{d}\right]^t + \frac{(d-1)\lambda^t+1}{d^k}\right) \nonumber\\
        &=& 1 - \epsilon\left[ \frac{d-1/d^{k-2}}{2(d-1)} + \frac{1}{2}(1+\frac{1}{d^{k-1}})t(1-\lambda) - O[(1-\frac{1}{d}+\frac{1}{d^{k-1}})t^2(1-\lambda)^2] + \cdots\right] \nonumber\\
        &\approx&  1 - \epsilon\left[1+(1-\lambda)t\right]/2\,,
    \end{eqnarray}
    where the expansion and approximation take into account two aspects for the large-scale multipartite entangled state respectively: (1) the second-largest eigenvalue of the standard QSV is, in general, close to unity; 
    (2) the dimension $d$ is on a large scale, so is the term $d^k$. 
    Therefore, the time complexity and sample complexity are
    \begin{eqnarray}
        M &=& \ln\delta^{-1}/\ln p_{\rm CN}^{-1} \approx \frac{2}{(1-\lambda)t+1}\epsilon^{-1}\ln\delta^{-1}\,,\\
        N &=& tM \approx \frac{2t}{(1-\lambda)t+1}\epsilon^{-1}\ln\delta^{-1}\,.
    \end{eqnarray}
    The derivation is similar to the fidelity that with
    \begin{eqnarray}
        \tr\left[\bigl(\Omega^{\otimes t}\otimes\ketbra{\psi}\otimes\openone\bigr) D_k\bigl(\ketbra{\psi}^{\otimes k}\bigr)\right] &=& 1\,,\\
        \tr\left[\bigl(\Omega^{\otimes t}\otimes\ketbra{\psi}\otimes\openone\bigr) D_k\bigl(\openone/d^k\bigr)\right] &=& \frac{1}{2}\left[\tr(\Omega\frac{\openone}{d})\right]^t \bra{\psi}\frac{\openone}{d}\ket{\psi}
        + \frac{1}{2}\tr\left[(\Omega\frac{\openone}{d})^t \ketbra{\psi} (\frac{\openone}{d})^{k-t}\right]\nonumber\\
        &=&\left(\left[\lambda+\frac{1-\lambda}{d}\right]^t \frac{1}{d} + \frac{1}{d^k}\right)/2\,,
    \end{eqnarray}
    we have 
    \begin{eqnarray}
        F = \left[(1-q)+q\left(\left[\lambda+\frac{1-\lambda}{d}\right]^t\frac{1}{d}+\frac{1}{d^k}\right)/2\right]/p_{\rm CN} \,,\qquad q=\frac{d\epsilon}{d-1}\,.
    \end{eqnarray}
    Then the infidelity is reduced to
    \begin{eqnarray}
        \epsilon' 
        &=& \epsilon \frac{d^k(d-1) \left(\lambda+\frac{1-\lambda}{d}\right)^t}{2d^k(d-1) -2d^{k+1}\epsilon +d^{k+1}\left(\lambda+\frac{1-\lambda}{d}\right)^t \epsilon} \nonumber\\
        &\approx& \frac{1}{2}\left(\lambda+\frac{1-\lambda}{d}\right)^t\epsilon + \frac{d}{4(d-1)}\left[2-\left(\lambda+\frac{1-\lambda}{d}\right)^t\right]\left(\lambda+\frac{1-\lambda}{d}\right)^t\epsilon^2 + O(\epsilon^3)\,.
    \end{eqnarray}
    Note that the eigenvalue is ${0<\lambda<1}$, thus ${a = \left(\lambda+\frac{1-\lambda}{d}\right)^t<1}$ such that the infidelity of unmeasured states decreases ${\epsilon' < \epsilon/2}$. 
    More specifically, this is an exponential decrease $\epsilon'\approx a^t\epsilon/2$ in terms of the unmeasured subset size $t$. 
    The improvement of the state fidelity could be significant for Bell states and large-scale multipartite entangled states; see Appendix~\ref{App:DetailsApp} for examples.
\end{proof}

\section{The global unitary control}\label{App:GU}
In principle, the most powerful adversary is able to apply global unitary controls on the whole ensemble, leading to a noisy state of the form
\begin{eqnarray}
    &&\eta_{\rm GU} = \ket{\phi}\bra{\phi}\,,\nonumber\\
    &&\ket{\phi} = \sqrt{1-\varepsilon}\ket{\psi}^{\otimes k} + \sqrt{\varepsilon}\ket{\phi'}\,,
\end{eqnarray}
where $\ket{\phi'}$ is a pure state resulting from a global unitary rotation of the ensemble of the target state $\ket{\psi}^{\otimes k}$. 
Each reduced state has the fidelity $\tr[(\ketbra{\psi}\otimes\openone) \ketbra{\phi}]=1-\epsilon$.
\begin{theorem}\label{thm:GUR}
    A target state \ket{\psi} can be verified by the collective strategy $\Pi_{k,t}$ under global unitary control within infidelity $\epsilon$ and confidence level $1-\delta$ via
    \begin{eqnarray}
        M = \ln\delta^{-1}/\ln p^{-1} \approx \frac{1}{(1-\lambda)t}\epsilon^{-1}\ln\delta^{-1}
    \end{eqnarray}
    rounds of testing, and ${N=tM}$ number of samples. Simultaneously, ${(k-t)M}$ copies of the output states $\sigma'$ with an increased infidelity of ${\epsilon'=\epsilon+(1-\lambda)t\epsilon^2}$ are produced.
\end{theorem}
\begin{proof}
Decompose the state $\ket{\phi'}$ under the basis $\{\ket{\psi},\ket{\psi_1^\perp},\ket{\psi_2^\perp},\cdots\}^{\otimes k}$ and by considering Eq.~\eqref{eq:p_iid}, we notice that the noisy pure state that can pass the collective scheme with the maximal passing probability should be all permutations of $\ket{\psi}^{\otimes(k-1)}\otimes\ket{\psi^\perp}$, where \ket{\psi^\perp} is the eigenstate corresponding to the eigenvalue $\lambda$, i.e., 
\begin{eqnarray}
    \ket{\phi'} = \frac{1}{\sqrt{k}}\sum_i \mathcal{P}_i\left\{\ket{\psi}^{\otimes k-1}\otimes \ket{\psi^\perp}\right\}\!,
\end{eqnarray}
where $\epsilon=\varepsilon/k$. 
Due to the permutation invariance of the worst case, we have 
\begin{eqnarray}
    D_k(\ket{\phi}\bra{\phi}) 
    = \frac{1}{4}\ket{\phi}\bra{\phi}
    + \frac{1}{4}S_k\ket{\phi}\bra{\phi}S_k^\dagger
    + \frac{1}{4}S_k \ket{\phi}\bra{\phi}
    + \frac{1}{4} \ket{\phi}\bra{\phi} S_k^\dagger
    = \ket{\phi}\bra{\phi}\,.
\end{eqnarray}
Hence, with the following derivations 
    \begin{eqnarray}
        \tr\left[\bigl(\Omega^{\otimes t}\otimes\openone\bigr) \ketbra{\psi}^{\otimes k}\right] &=& 1\,,\\
        \tr\left[\bigl(\Omega^{\otimes t}\otimes\openone\bigr) \ketbra{\phi'}\right] &=& \frac{k-t+t\lambda}{k} \,,\\
        \tr\left[\bigl(\Omega^{\otimes t}\otimes\openone\bigr) \ket{\psi}^{\otimes k}\bra{\phi'}\right] &=& 0\,,
    \end{eqnarray}
the passing probability after the global unitary control becomes
    \begin{eqnarray}\label{eq:p_GU}
        p_{\rm GU} = \frac{1}{C_k^t}\sum_i\tr\Bigl[\cP_i\left\{\bigl(\Omega^{\otimes t}\otimes\openone\bigr)\right\}D_k\bigl(\ket{\phi}\bra{\phi}\bigr) \Bigr] = 1 - \varepsilon + \varepsilon\frac{k-t+t\lambda}{k} = 1 - \epsilon(1-\lambda)t\,.
    \end{eqnarray}
Thus we need 
\begin{eqnarray}
    M \approx \frac{1}{(1-\lambda)t}\epsilon^{-1}
    \ln\delta^{-1}-\frac{1}{2}\ln\delta^{-1}
\end{eqnarray}
rounds of testing to verify the target state within infidelity $\epsilon$ and confidence level $1-\delta$.
The fidelity of the unmeasured states is 
    \begin{eqnarray}
        F = \tr\left[\ketbra{\psi}\otimes\openone\cdot\Pi_{(k,t)}\bigl(\eta_{\rm GU}\bigr)\right]/p_{\rm GU}
        = \left\{1-k\epsilon+\epsilon(k-t-1+t\lambda)\right\}/p_{\rm GU}
        = 1 - \frac{\epsilon}{1-\epsilon(1-\lambda)t}\,,
    \end{eqnarray}
    so the infidelity is increased to
    \begin{eqnarray}
        \epsilon' = \epsilon\left[1 + (1-\lambda)t\epsilon + (1-\lambda)^2t^2\epsilon^2 + O(\epsilon^3)\right] \approx \epsilon + (1-\lambda)t\epsilon^2\,.
    \end{eqnarray}
\end{proof}

One immediately notices that, different from the cases of independent or global white noise, the fidelity of the unmeasured states has decayed, exactly the reason why the adversary that is able to control the entire quantum system can be tricky, similar to the discussion in Ref.~\cite{Zhu.Hayashi2019c}. 
Therefore, in order to guarantee enough quality of the unmeasured states for any sequential tasks, we need to consume more resource for verification. 
The following corollary, directly derived from Theorem~\ref{thm:GUR}, is thus more useful in practice.
\begin{corollary}
    The collective strategy $\Pi_{k,t}$ produces $(k-t)M'$ unmeasured states within the infidelity $\epsilon$ and confidence level $1-\delta$ under global unitary control via
    \begin{eqnarray}
        M' 
        = \frac{\ln\delta^{-1}}{\ln\left[1 - \frac{\epsilon}{1+\epsilon(1-\lambda)t} (1-\lambda)t\right]^{-1} } 
        \approx  \frac{1}{(1-\lambda)t}\epsilon^{-1}\ln\delta^{-1} + \frac{1}{2}\ln\delta^{-1}
    \end{eqnarray}
    rounds of testing, and $N'=tM'$ number of sample consumption.
\end{corollary}
Notably, the additional cost for guaranteeing the unmeasured states is only dominated by the confidence level
\begin{eqnarray}
    M' - M = \ln\delta^{-1} \left[1 + \frac{1}{12}(1-\lambda)^2t^2\epsilon^2+O(\epsilon^4)\right]\!,
\end{eqnarray}
which is relatively small as $(M'-M)/M = (1-\lambda)t\epsilon + O(\epsilon^2)$ for the typical infidelity of $\epsilon\leq5\%$.
For high-accuracy scenarios where ${\epsilon\to 0}$, and ${\epsilon'\to 0}$ as well, the additional resource cost can be safely ignored. 
For a comparison, Ref.~\cite{Zhu.Hayashi2019c} considers a more powerful (but unrealistic) adversary who can control the whole system, not just the whole $k$-sized ensemble, the sample complexity is at least
\begin{eqnarray}
    N_{\rm adv} 
    \geq \frac{\ln\delta}{\ln\lambda} + \frac{\ln\delta}{\ln\lambda}\frac{1-\epsilon}{\lambda\epsilon}
    \approx (\frac{1}{1-\lambda}+\E - 1)\epsilon^{-1}\ln\delta^{-1}\,,
\end{eqnarray}
where the additional cost dominated by the infidelity is significant which cannot be ignored.
Though this is not a completely fair comparison, the collective scheme achieves an improvement on the sample complexity by
\begin{eqnarray}
 N_{\rm adv} - N' \approx \left(\E -2\right)\epsilon^{-1}\ln\delta^{-1} - \frac{t}{2}\ln\delta^{-1}\,.
\end{eqnarray}

\section{Details of the applications}\label{App:DetailsApp}
In the main text, we have briefly outlined the efficiency of our collective scheme, $\Pi_{k,1}$, for the verification of Bell state and Dicke state, showing that it outperforms the optimal global verification. This improvement is especially pronounced for the 100-qubit Dicke state, as illustrated in Fig.~\ref{fig:p}. 
In the following, we provide further details and additional results.

\begin{figure}[t]
  \includegraphics[width=0.75\columnwidth]{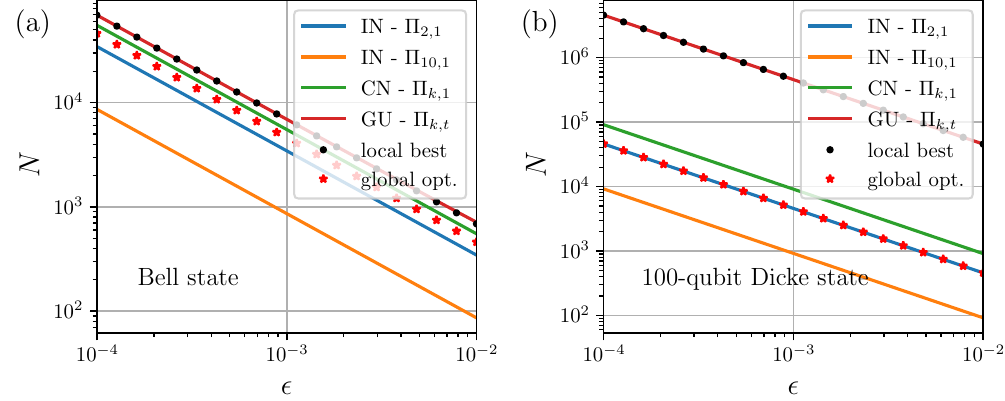}
  \caption{Sample complexity for verifying (a) a Bell state and (b) a $100$-qubit Dicke state using the collective strategies under independent white noise with schemes $\Pi_{2,1}$ (blue) and $\Pi_{10,1}$ (orange), under global white noise with the scheme $\Pi_{k,1}$ (green), and under global unitary control with the scheme $\Pi_{k,t}$ (brown) as compared to the best known local scheme \cite{Pallister.etal2018, Li.etal2020a} (black dot) and the optimal global scheme (red star). The confidence level is set to $1-\delta=99\%$.}
  \label{fig:noise}
\end{figure}
\subsection{Comparison of different noises}
First, let's consider another type of noise, which is the global unitary control as discussed in Appendix~\ref{App:GU}. 
In Fig.~\ref{fig:noise}, we compare the sample complexity for verifying a Bell state and a $100$-qubit Dicke state under independent white noise, global white noise, and global unitary control. 
Beyond the cases discussed in the main text,  our results here show that the collective scheme beats the adversary who can control the whole ensemble, guaranteeing the fidelity of an unmeasured state with the same sample complexity as the standard QSV with no adversary. 
On the other hand, the higher complexity under global unitary control enables us to differentiate powerful adversary attacks from natural white noise.

\subsection{Influence of the ensemble size and measured subset size}
\begin{figure}[t]
  \includegraphics[width=0.75\columnwidth]{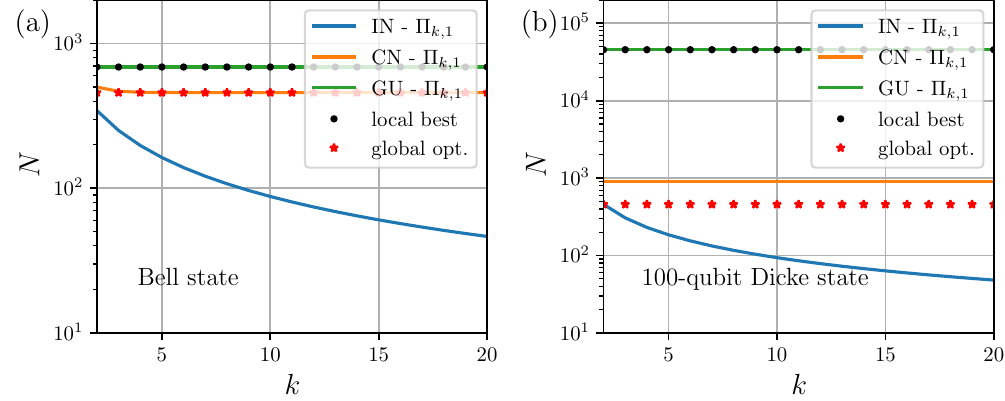}
  \caption{Sample complexity for verifying (a) a Bell state and (b) a $100$-qubit Dicke state using the collective schemes $\Pi_{k,1}$ with varying ensemble size $k$. 
  The independent white noise (blue), global white noise (orange), and global unitary control (green) are all considered, as compared to the best known local scheme \cite{Pallister.etal2018, Li.etal2020a} (black dot) and the optimal global scheme (red star). The fidelity is set to $1-\epsilon=99\%$ and the confidence level is $1-\delta=99\%$.}
  \label{fig:k}
\end{figure}
\begin{figure}[t]
  \includegraphics[width=0.75\columnwidth]{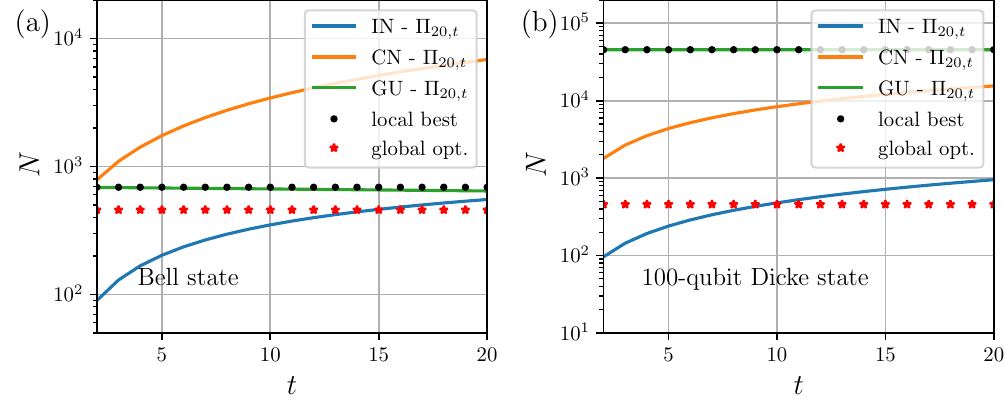}
  \caption{Sample complexity for verifying (a) a Bell state and (b) a $100$-qubit Dicke state using the collective scheme $\Pi_{20,t}$ with varying measured subset size $t$. 
  The independent white noise (blue), global white noise (orange), and global unitary control (green) are all considered, as compared to the best known local scheme \cite{Pallister.etal2018, Li.etal2020a} (black dot) and the optimal global scheme (red star). The fidelity is set to $1-\epsilon=99\%$ and the confidence level is $1-\delta=99\%$.}
  \label{fig:t}
\end{figure}

The ensemble size and measured subset have different influences on the resource consumption for verifying the entangled states suffering from different types of noise.
Based on Theorems~\ref{thm:main}-\ref{thm:GUR}, a larger ensemble can improve the efficiency infinitely in terms of both the time and sample consumption for verifying an entangled state under independent white noise. 
But for states under global white noise or global unitary control, the consumption becomes independent of the ensemble size. 
One might realize that all the resource complexities are approximated under corresponding assumptions and wonder under what practical conditions these approximations hold.
In Fig.~\ref{fig:k}, we present precise numerical calculations of the sample complexity for verifying a Bell state and a $100$-qubit Dicke state, considering various ensemble sizes $k$.
The results are based on the passing probabilities as in Eqs.~\eqref{eq:p_IN}, \eqref{eq:p_CN}, and \eqref{eq:p_GU} without approximations.
The fidelity and confidence level are both set to a typical value of $99\%$. 
Overall, the numerical results align well with the theoretical predictions.
Exception occurs for the case of verifying a Bell state under global white noise, namely the orange curve in Fig.~\ref{fig:k}(a), it does exhibit a slight dependence of sample complexity on the ensemble size $k$.
The dependence can be ignored for larger ensembles (e.g. ${k\geq 4}$ for a Bell state) and for the case of a $100$-qubit Bell state, which has a large dimension $d=2^{100}$. 

In Fig.~\ref{fig:t}, we perform similar precise numerical calculations of the sample complexity with varying unmeasured subset size $t$. 
If the entangled state suffers from independent or global white noise, a larger measured subset can worsen the sample complexity. 
However, for the case of global unitary control (the green curve in Fig.~\ref{fig:k}(a)), the sample complexity can be slightly decreased with the increasing unmeasured subset size $t$, with a ratio of  $t/[2(1-\lambda)\epsilon]$.
This improvement can be ignored for the high-accuracy scenario as ${t\ll\epsilon^{-1}}$ or the complex multipartite entangled state with ${t\ll(1-\lambda)^{-1}}$, which is the case in Fig.~\ref{fig:t}(b).

\subsection{Enhancement of unmeasured states}
\begin{figure}[t]
  \includegraphics[width=0.95\columnwidth]{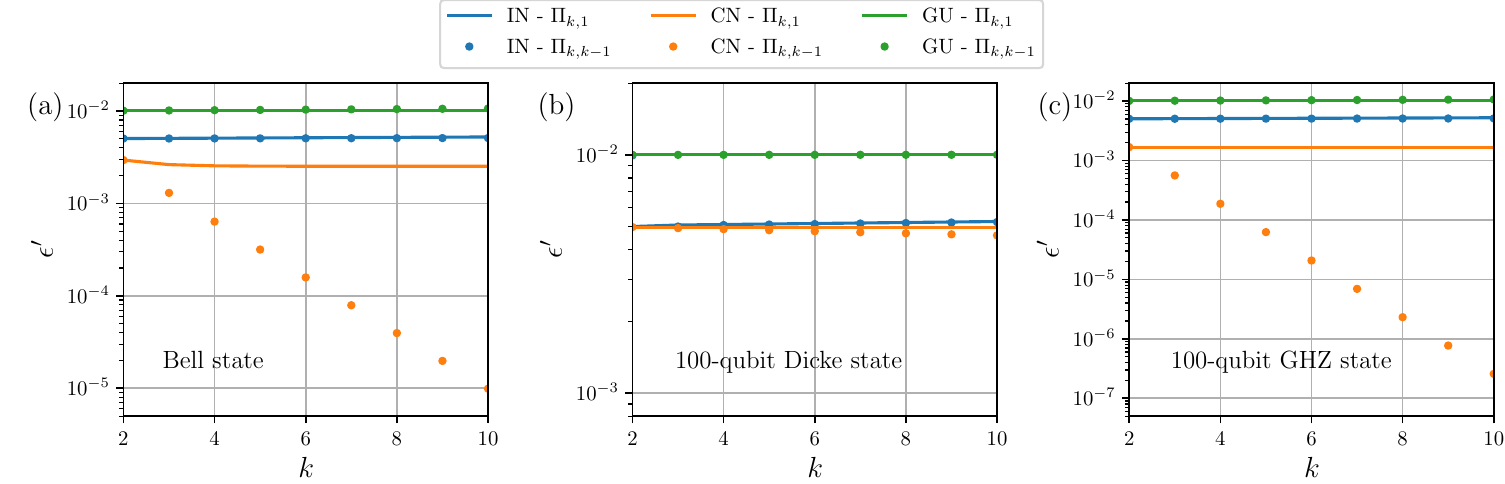}
  \caption{Decreased infidelity of the unmeasured states when verifying (a) a Bell state, (b) a $100$-qubit Dicke state, and (c) a $100$-qubit GHZ state using the collective scheme $\Pi_{k,1}$ under independent white noise (blue), global white noise (orange), and global unitary control (green). The dotted line is the case for more measurements as $\Pi_{k,k-1}$. The original infidelity is set to $\epsilon=0.01$.}
  \label{fig:eps}
\end{figure}

One notable advantage of our collective QSV scheme is the enhancement of unmeasured states during the verification procedure. 
Taking the Bell state, a $100$-qubit Dicke state, and a $100$-qubit GHZ state as examples, we illustrate the explicit numerical results in Fig.~\ref{fig:eps} for one of the unmeasured copies with the original infidelity ${\epsilon=1\%}$. 
Matching the analysis of our theorems, for independent white noise, our collective scheme can mitigate the noise to be half, and for global white noise, the scheme performs even better. 
Even if a powerful adversary can control the entire ensemble, the unmeasured states can still be verified, only with negligible infidelity increasing. 
Fully utilizing the ensemble with collective scheme $\Pi_{k,k-1}$ can make the enhancement of unmeasured states better. 
For entangled states with high symmetry (thus $\lambda$ being independent of the system size), e.g., Bell state and GHZ state, the global white noise can be exponentially decreased with the increased size of the ensemble, and thus can be ignored after passing the verification. 

According to Eq.~\eqref{eq:Dk}, the unmeasured copies of states can be separable as a tensor product only for two cases, namely that either all the copies are the target state or i.i.d. pure noisy states.
In general, however, it is important to note that each round of testing introduces certain correlation among the unmeasured ${k-t}$ copies of the states.
In most cases, subsequent tasks rely on the independent use of these verified states.
However, caution is required in two types of scenarios where correlation in the unmeasured copies might have an influence: 
(1) if the subsequent task collectively utilizes multiple copies of the states;
(2) if time correlation of the outcomes, rather than their statistical distribution, is of interest.
Under these circumstances, it is prudent to discard all but one unmeasured copy per round to ensure its independent use.

\subsection{The tradeoff}
Let's consider a task for verifying a target state within infidelity ${\epsilon=1\%}$ and confidence level ${1-\delta=99\%}$  under the independent white noise. 
For the worst consideration, of which ${1-\lambda \to 0}$ for an arbitrary target state, the collective scheme $\Pi_{k,t}$ needs $\lceil \frac{2}{k}\epsilon^{-1}\ln\delta^{-1}\rceil$ rounds of tests and consumes $t\lceil \frac{2}{k}\epsilon^{-1}\ln\delta^{-1}\rceil$ additional samples. 
Therefore, the larger ensembles are utilized, the faster the verification is; the less samples are measured in each round, the less the sample consumption is. 
For example, up to $461$ rounds of the tests are needed to verify an arbitrary target state with $\Pi_{2,1}$, and the same amount of additional samples are consumed, which is the worst case. With the scheme $\Pi_{10,1}$, it only needs $93$ rounds of the tests.
However, when considering specific target states, the sample and time consumption vary according to $\lambda$, leading to a tradeoff between the time consumption and the sample consumption of the verification procedure.  
For Bell states, the waiting time is $87$ rounds of the test with the scheme $\Pi_{10,1}$, consuming $87$ samples. The more samples are measured, the faster the verification is. 
The scheme $\Pi_{10,9}$ needs to wait for $58$ rounds of the tests, while consuming $58\times9=522$ samples, which is more than the sample consumption with the scheme $\Pi_{10,1}$. 

As we partially discussed in the main text, let's consider a task demanding $1024$ copies of a target state, verified within infidelity ${\epsilon=1\%}$ and confidence level ${1-\delta=99\%}$ under the independent white noise.
As the requirement of $1024$ copies (a typical data size in quantum computation) is larger than the worst-case which requires $461$ rounds of the tests, the task can always be achieved with a better verification. 
Using larger ensembles can speed up the verification procedure, while the verified infidelity threshold increases. 
For example, in order to provide $1024$ copies of a verified target state, the simplest setup $\Pi_{2,1}$ needs to consume another $1024$ copies and take $1024$ rounds of tests, verifying these target states within infidelity $0.45\%$ under confidence level ${1-\delta=99\%}$. 
With larger ensembles, the scheme $\Pi_{9,1}$ consumes only $128$ additional samples
and needs $128$ rounds of tests, verifying these target states within infidelity $0.80\%$ under confidence level ${1-\delta=99\%}$. 

Considering specific target states, the size of the measured subset leads to a tradeoff between the consumption and the accuracy of the verification procedure due to the $\lambda$ dependence. 
For the Bell state, the scheme $\Pi_{10,1}$ consumes $114$ additional samples
and needs $114$ rounds of tests, for verifying the target state within infidelity $0.76\%$ under confidence level ${1-\delta=99\%}$. 
While the scheme $\Pi_{10,5}$ which measures more samples in each round, consumes more with $1025$ additional samples and $205$ rounds of tests, for verifying the target state better within infidelity $0.34\%$ under confidence level $1-\delta=99\%$.

\section{Advantages of our collective scheme as compared to other related works}\label{App:Comp}
In this section, we compare our collective scheme with other related works \cite{miguel-ramiro_collective_2022, chen_memoryQSV_2025} on verifying entangled states with collective strategies, emphasizing further the unique properties and practical advantages of our approach.

\subsection{Universality for verifying arbitrary entangled states}
A notable superiority of our scheme lies in its universality for verifying arbitrary entangled states, which relies on two key factors : 
(1) The SWAP projection is capable of purifying arbitrary quantum states \cite{ricci_experimental_2004}; (2) Valid verification protocols exist for all entangled states. 
In contrast, Ref.~\cite{miguel-ramiro_collective_2022} proposes a collective strategy based on entanglement distillation operations, which achieves low sample complexity but is applicable to Bell states and (generalized) GHZ states only. 
Ref.~\cite{chen_memoryQSV_2025} demonstrates an efficient two-copy collective strategy for graph states by leveraging the connection between graph codes and their parity codes. While this approach is possible to be extended to multi-copy by increasing the local dimensions, it remains specific to graph states only.

\subsection{Arbitrarily high improvement of verification}
In general, a verification task with collective strategies can be described as follows: 
A quantum device generates $k$ copies of an unknown state per round, which undergo collective operations according to a specific strategy. 
Of those, $t(\leq k)$ copies are measured while the remaining are either discarded or used later.
This procedure is repeated for $M$ rounds, consuming ${N=tM}$ copies of the states in total.
Based on certain passing rules, with confidence level ${1-\delta}$, the state is guaranteed to be close to the target state with a fidelity better than ${1-\epsilon}$. 

The collective strategy proposed in Ref.~\cite{miguel-ramiro_collective_2022}, though limited to verifying Bell states only, is claimed to have an exponential improvement of verification.
Specifically, the entangled operations, which we call the error number gates, are sequentially applied on $m$ Bell states and an ancillary maximally entangled state with dimension ${d=m+1}$.
It is the high-dimensional maximally entangled state, which can be constructed by embedding $\log_2(m+1)$ Bell states, needs to be measured and thus destroyed in the end.
The passing probability for the noisy Bell state in each round is ${p_{\rm ED, Bell}=1-m\epsilon/2+O(\epsilon^2)}$.
Therefore, in each round, the so-called exponentially improved scheme of Ref.~\cite{miguel-ramiro_collective_2022} needs ${k_{\rm ED, Bell}=2^t-1+t}$ copies of Bell states of which $t$ copies are measured.
Hence, the Bell state can be verified within infidelity $\epsilon$ and confidence level $1-\delta$ via 
\begin{eqnarray}
    M_{\rm ED, Bell} = \frac{2}{2^t-1}\epsilon^{-1}\ln\delta^{-1}
\end{eqnarray}
rounds of testing, and $N_{\rm ED, Bell}=tM_{\rm ED, Bell}$ number of sample consumption.

Considering Theorem~\ref{thm:main}, if we use the same setting as that of Ref~\cite{miguel-ramiro_collective_2022} by preparing ${k=k_{\rm ED, Bell}}$ copies then measuring $t$ of them, our collective scheme can also beat the one of Ref~\cite{miguel-ramiro_collective_2022}, i.e.,
\begin{eqnarray}
    M(k=k_{\rm ED, Bell},t) = \frac{2}{2^t-1+t+(1-\lambda)t}\epsilon^{-1}\ln\delta^{-1} < M_{\rm ED, Bell}\,.
\end{eqnarray}
More importantly, in our scheme, the number of all quantum states $k$ and the number of measured ones $t$ are free to choose, without any restrictions.
Therefore, with more and more samples in each round, we can still measure just one (or a constant number) of them, reaching the arbitrarily high efficiency, which can be better than the exponential improvement of Ref.~\cite{miguel-ramiro_collective_2022}.
Moreover, if the ensemble is sufficiently large, the resource consumed would asymptotically decrease to zero, meaning that the verification can be completed within just one round using a single copy of the state. 

\subsection{Feasibility for experimental realization}
\begin{figure}[t]
  \includegraphics[width=0.55\columnwidth]{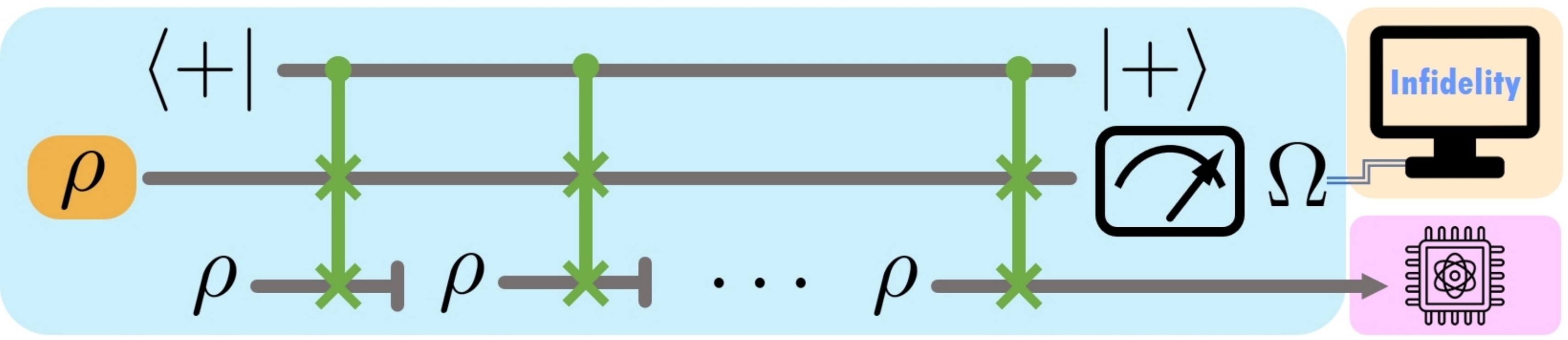}
  \caption{Sketch of the most memory-saving collective verification scheme. It only requires the quantum memory to store one copy of the quantum state (orange shadow) and one ancillary qubit, which are measured while all the other copies in the ensemble are discarded after the collective operations. The last unmeasured copy can be used for subsequent tasks.}
  \label{fig9}
\end{figure}

In practice, the quantum states produced by quantum devices are considered inexpensive, and also implementing collective operations, though challenging, has become feasible at intermediate scales. 
The most expensive and limited resource is the quantum memory or quantum register, which is required to store the quantum states for collective operations prior to measurements. 

The freedom to choose any $t$ copies of the states to measure from the $k$-ensemble offers another significant advantage of our scheme, namely the efficient use of quantum memory. 
Even with quantum memory sufficient for only one copy of the quantum state, our scheme achieves arbitrarily high efficiency, surpassing the optimal global verification. 
The sketch of such most memory-saving collective scheme is illustrated in Fig.~\ref{fig9}. 
It is implemented by sequentially interacting one registered quantum state with others, of which the registered one is measured and the others are discarded. This variation actually realizes another SWAP projection with the inverse of the cyclic permutation $S_k$.
Due to symmetry with the SWAP projection $D_k$, it still follows Theorem~\ref{thm:main}, and
the time and sample complexity of this collective verification is given by
\begin{eqnarray}
    N = M = \ln\delta^{-1}/\ln p^{-1} \approx \frac{2}{(1-\lambda)+k}\epsilon^{-1}\ln\delta^{-1}\,,
\end{eqnarray}
which outperforms the optimal global verification when $k\geq2$.
In each round, without additional quantum memory, $k-2$ copies of the unmeasured states are discarded, 
while the final unmeasured one state with a smaller infidelity of $\epsilon/2$ is produced.
Considering the worst case that a powerful adversary is able to control the whole ensemble, according to Theorem~\ref{thm:GUR}, the time and sample complexity is
\begin{eqnarray}
    N' = M' \approx \frac{1}{1-\lambda}\epsilon^{-1}\ln\delta^{-1}\,.
\end{eqnarray}

In contrast, the collective schemes in Ref.~\cite{miguel-ramiro_collective_2022} and Ref.~\cite{chen_memoryQSV_2025} demand significantly larger quantum memory for multi-copy strategies. 
Consider the verification task for Bell states, of which both of them can realize. 
In each round, for $k$ copies of input states, the scheme of Ref.~\cite{miguel-ramiro_collective_2022} requires at least $O(\log(k))$ quantum registers to construct a high-dimensional entangled ancilla, which is subsequently measured. 
The scheme of Ref.~\cite{chen_memoryQSV_2025} treats all $k$-copy Bell states as a high-dimensional GHZ state, necessitating memory for all copies. 
While both schemes deliver high performance when extended, their reliance on large quantum memory sets a higher experimental threshold. 
A detailed comparison of the resource requirements for our memory-saving collective scheme and those of Refs.~\cite{Pallister.etal2018, miguel-ramiro_collective_2022, chen_memoryQSV_2025} is provided in Table~\ref{tab:resource}.
\begin{table}[h]
    \centering
    \caption{Resource needed for verifying Bell states by different schemes.}
    \begin{tabular}{l|c|c|c|c}
    \hline
                                      &Time ($\epsilon^{-1}\ln\delta^{-1}$)   &Sample ($\epsilon^{-1}\ln\delta^{-1}$)     &Operations     &Memory (qubit) \\
    \hline
    Our scheme with i.i.d. white noise      &$\frac{6}{3k+2}$ \footnote{This complexity is computed based on the optimal local QSV \cite{Pallister.etal2018}, and one can also consider the sequential verification \cite{Liu.etal2020b} with the complexity being $\frac{2}{k}$.}
                                                                    &$\frac{6}{3k+2}$       &$O(k)$ Fredkin gates                                                     &$2+1$
                                                                    \\ 
    \hline
    Our scheme under worst case             &$\frac{3}{2}$          &$\frac{3}{2}$          &$O(k)$ Fredkin gates                                                     &$2+1$\\ 
    \hline
    Ref.~\cite{miguel-ramiro_collective_2022} 
    with i.i.d. white noise                 &$\frac{2}{2^t-1}$ \footnote{For the scheme in Ref.~\cite{miguel-ramiro_collective_2022}, in each round, we need $k\!=\!(2^t\!-\!1)\!+\!t$ copies of the states in total, depending on the size of the quantum memory.}     
                                                                    &$\frac{2t}{2^t-1}$      &$O(2^t)$ controlled-$X_{2^t\!-\!1}$ gates                           &$2t$\\ 
    \hline
    Ref.~\cite{chen_memoryQSV_2025} 
    under worst case                         &$\frac{2^k}{2^k-1}$   &$\frac{k2^k}{2^k-1}$   &$O(2^{k})$ $Z_{2^k}$ gates and $X_{2^k}$ gates \cite{Li.etal2020b}    &$2k$\\
    \hline
    Ref.~\cite{Pallister.etal2018} without collective scheme         &$\frac{3}{2}$ \footnote{The complexity of the scheme in Ref.~\cite{Pallister.etal2018} is based on the consideration of all possible local noise. For the worst case here, Ref.~\cite{Zhu.Hayashi2019d} gives the scaling as Euler's number $\mathrm e \approx 2.718$.}      &$\frac{3}{2}$      &$2$ $X, Y, Z$ projections ($6$ in total)   &$0$\\
    \hline
    \end{tabular}
    \label{tab:resource}
\end{table}
\subsection{Distributed construction of our collective scheme}
\begin{figure}[t]
  \includegraphics[width=0.55\columnwidth]{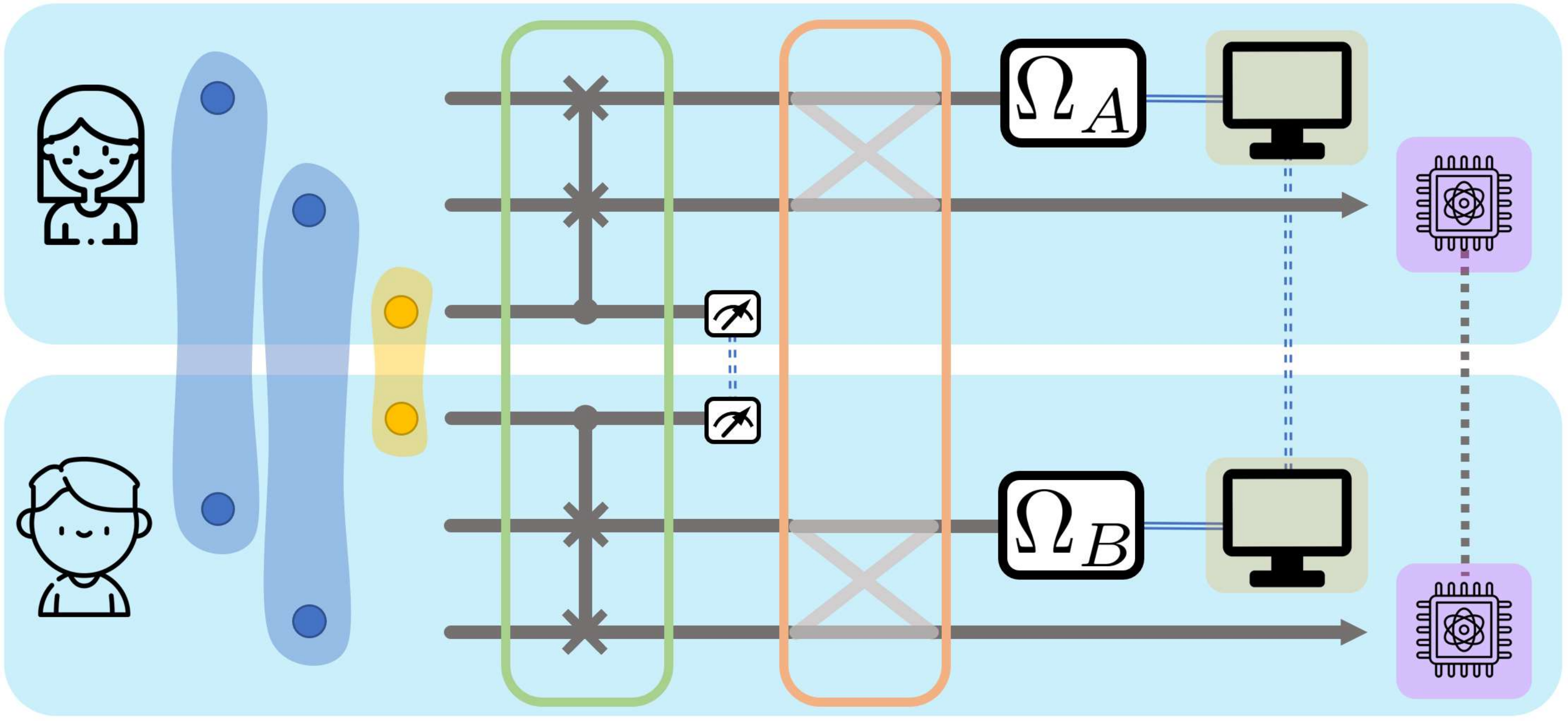}
  \caption{Sketch of the distributed construction of the collective verification scheme with an entangled ancilla for two-sized ensembles. The SWAP projections are local and controlled by an ancillary Bell state. The measurement on the ancilla is the Bell projection, which can be realized locally assisted with classical communications.}
  \label{fig:local}
\end{figure}
One notable feature of our collective scheme is the use of a single ancillary qubit. While this approach conserves resources, it poses challenges for constructing distributed verification protocols for multipartite entangled states. However, this issue can be tackled by employing a maximally entangled state of qubits instead of a single ancillary qubit. 
As the sketch shown in Fig.~\ref{fig:local}, for the bipartite case, we have the equivalence 
\begin{eqnarray}
    D_k(\rho) 
    &=& \tr_{{\rm anA},{\rm anB}}\Bigl[\ket{\Phi}\bra{\Phi}_{{\rm anA},{\rm anB}}\otimes\openone_{A1,A2,B1,B2} \nonumber\\
    &&\qquad \cdot cS_{{\rm anA},A1,A2}\otimes cS_{{\rm anB},B1,B2} \bigl(\ket{\Phi}\bra{\Phi}_{{\rm anA},{\rm anB}}\otimes \rho_{A1,B1}\otimes\rho_{A2,B2} \bigr) cS_{{\rm anA},A1,A2}^\dag\otimes cS_{{\rm anB},B1,B2}^\dag\Bigr]\nonumber\\
    &=& \tr_{{\rm anA},{\rm anB}}\Bigl[(\ket{+}\bra{+}_{\rm anA}\otimes\ket{+}\bra{+}_{\rm anB}+\ket{-}\bra{-}_{\rm anA}\otimes\ket{-}\bra{-}_{\rm anB})\otimes\openone_{A1,A2,B1,B2} \nonumber\\
    &&\qquad \cdot cS_{{\rm anA},A1,A2}\otimes cS_{{\rm anB},B1,B2} \bigl(\ket{\Phi}\bra{\Phi}_{{\rm anA},{\rm anB}}\otimes \rho_{A1,B1}\otimes\rho_{A2,B2} \bigr) cS_{{\rm anA},A1,A2}^\dag\otimes cS_{{\rm anB},B1,B2}^\dag\Bigr] \nonumber\\
    &=& \frac{\openone+S_A\otimes S_B}{2}\rho \left(\frac{\openone+S_A \otimes S_B}{2}\right)^\dag\!,
\end{eqnarray}
where the ancilla ${\ket{\Phi}=(\ket{00}+\ket{11})/\sqrt{2}}$ is the Bell state.
Note that the controlled SWAP can be operated locally, and the measurements on the ancillary Bell state can also be local with the help of classical communications. 
This technique is similar to that of Refs.~\cite{miguel-ramiro_collective_2022, chen_memoryQSV_2025}, but no higher-dimensional local systems, or equivalently multiple ancillary Bell states, are needed.

In practice, the initialization of an ancillary qubit is convenient, but rather difficult for an entangled ancillary pair.
Luckily, in our distributed construction, the Bell state is not necessarily to be perfect. 
For the bipartite distributed verification task, we consider a noisy ancillary Bell state as 
\begin{eqnarray}
    \rho_{\rm ancilla} = \bigl(1-\gamma\bigr)\ket{\Phi}\bra{\Phi} + \gamma\frac{\openone}{4}\,.
\end{eqnarray}
Due to linearity, with the noisy ancilla, we can realize a noisy SWAP projection as 
\begin{eqnarray}
    D'_k(\rho) = \left(1-\frac{\gamma}{2}\right)D_k(\rho)\,.
\end{eqnarray}

For the case of distributed verification of Bell states, with the ancillary Bell state being one copy of the (possibly) noisy quantum state (${\gamma=q=d\epsilon/(d-1)}$), 
the passing probability of measuring $t$ samples in the $k$-sample ensemble by the collective QSV is
\begin{eqnarray}
        p' = \frac1{C_{k-1}^t}\sum_i\tr\left[\cP_i\left\{\bigl(\Omega^{\otimes t}\otimes\openone\bigr)\right\}D'_{k-1}\bigl(\sigma^{\otimes {k-1}}\bigr)\right] = \left\{\frac{1}{2}\left[\tr(\Omega\sigma)\right]^t 
        + \frac{1}{2}\tr\left[(\Omega\sigma)^t\sigma^{k-1-t}\right]\right\}\left(1-\frac{\gamma}{2}\right).           
    \end{eqnarray}
Note that, because of the ancillary Bell state, the SWAP projection in this case is applied on $(k-1)$ copies.
Therefore, for Bell states, we have
    \begin{eqnarray}
        p'_{\rm IN} = \frac{1}{2}\left(1-\frac{2}{3}\epsilon\right) \left[(1 - \epsilon + \lambda \epsilon)^t + (1-\epsilon)^{k-1} + \frac{\epsilon^{k-1}\lambda^t}{3^{k-2}}\right]
        \approx 1 - \frac{3k+3(1-\lambda)t+1}{6}\epsilon + O(\epsilon^2)\,.
    \end{eqnarray}
One can immediately find that the passing probability of the noisy case $p'_{\rm IN}$ is actually smaller than that of the perfect case $p_{\rm IN}$, meaning that using a noisy Bell state as the ancilla makes the Bell states easier to verify. 
We only need 
\begin{eqnarray}
    M = \ln\delta^{-1}/\ln p^{-1} \approx \frac{2}{(1-\lambda)t+k+\frac{1}{3}}\epsilon^{-1}\ln\delta^{-1}
\end{eqnarray}
rounds of testing, and ${N = (t+1)M}$ number of samples. 

The better performance in sample consumption is not surprising.
This is a feature coming from the framework of quantum state verification first introduced in Ref.~\cite{Pallister.etal2018}, namely the verified infidelity $\epsilon$ accounts for all possible noise. 
Thus, we can always guarantee that the unknown states have a fidelity with the target state better than $1-\epsilon$ in the presence of noise. 
In other words, if we know the noise of gates, the verification will consume fewer samples to guarantee the fidelity of unknown states. 
Therefore, in general, the noisy gates will not make the efficiency worse in a verification task, of which our collective verification scheme retains as well.


\begin{thebibliography}{64}%
\makeatletter
\providecommand \@ifxundefined [1]{%
 \@ifx{#1\undefined}
}%
\providecommand \@ifnum [1]{%
 \ifnum #1\expandafter \@firstoftwo
 \else \expandafter \@secondoftwo
 \fi
}%
\providecommand \@ifx [1]{%
 \ifx #1\expandafter \@firstoftwo
 \else \expandafter \@secondoftwo
 \fi
}%
\providecommand \natexlab [1]{#1}%
\providecommand \enquote  [1]{``#1''}%
\providecommand \bibnamefont  [1]{#1}%
\providecommand \bibfnamefont [1]{#1}%
\providecommand \citenamefont [1]{#1}%
\providecommand \href@noop [0]{\@secondoftwo}%
\providecommand \href [0]{\begingroup \@sanitize@url \@href}%
\providecommand \@href[1]{\@@startlink{#1}\@@href}%
\providecommand \@@href[1]{\endgroup#1\@@endlink}%
\providecommand \@sanitize@url [0]{\catcode `\\12\catcode `\$12\catcode
  `\&12\catcode `\#12\catcode `\^12\catcode `\_12\catcode `\%12\relax}%
\providecommand \@@startlink[1]{}%
\providecommand \@@endlink[0]{}%
\providecommand \url  [0]{\begingroup\@sanitize@url \@url }%
\providecommand \@url [1]{\endgroup\@href {#1}{\urlprefix }}%
\providecommand \urlprefix  [0]{URL }%
\providecommand \Eprint [0]{\href }%
\providecommand \doibase [0]{https://doi.org/}%
\providecommand \selectlanguage [0]{\@gobble}%
\providecommand \bibinfo  [0]{\@secondoftwo}%
\providecommand \bibfield  [0]{\@secondoftwo}%
\providecommand \translation [1]{[#1]}%
\providecommand \BibitemOpen [0]{}%
\providecommand \bibitemStop [0]{}%
\providecommand \bibitemNoStop [0]{.\EOS\space}%
\providecommand \EOS [0]{\spacefactor3000\relax}%
\providecommand \BibitemShut  [1]{\csname bibitem#1\endcsname}%
\let\auto@bib@innerbib\@empty
\bibitem [{\citenamefont {Bennett}\ \emph {et~al.}(1993)\citenamefont
  {Bennett}, \citenamefont {Brassard}, \citenamefont {Cr\'epeau}, \citenamefont
  {Jozsa}, \citenamefont {Peres},\ and\ \citenamefont
  {Wootters}}]{Bennett.etal1993}%
  \BibitemOpen
  \bibfield  {author} {\bibinfo {author} {\bibfnamefont {C.~H.}\ \bibnamefont
  {Bennett}}, \bibinfo {author} {\bibfnamefont {G.}~\bibnamefont {Brassard}},
  \bibinfo {author} {\bibfnamefont {C.}~\bibnamefont {Cr\'epeau}}, \bibinfo
  {author} {\bibfnamefont {R.}~\bibnamefont {Jozsa}}, \bibinfo {author}
  {\bibfnamefont {A.}~\bibnamefont {Peres}},\ and\ \bibinfo {author}
  {\bibfnamefont {W.~K.}\ \bibnamefont {Wootters}},\ }\bibfield  {title}
  {\bibinfo {title} {Teleporting an unknown quantum state via dual classical
  and {E}instein-{P}odolsky-{R}osen channels},\ }\href
  {https://doi.org/10.1103/PhysRevLett.70.1895} {\bibfield  {journal} {\bibinfo
   {journal} {Phys. Rev. Lett.}\ }\textbf {\bibinfo {volume} {70}},\ \bibinfo
  {pages} {1895} (\bibinfo {year} {1993})}\BibitemShut {NoStop}%
\bibitem [{\citenamefont {Ekert}(1991)}]{Ekert1991}%
  \BibitemOpen
  \bibfield  {author} {\bibinfo {author} {\bibfnamefont {A.~K.}\ \bibnamefont
  {Ekert}},\ }\bibfield  {title} {\bibinfo {title} {Quantum cryptography based
  on {B}ell's theorem},\ }\href {https://doi.org/10.1103/PhysRevLett.67.661}
  {\bibfield  {journal} {\bibinfo  {journal} {Phys. Rev. Lett.}\ }\textbf
  {\bibinfo {volume} {67}},\ \bibinfo {pages} {661} (\bibinfo {year}
  {1991})}\BibitemShut {NoStop}%
\bibitem [{\citenamefont {Bennett}\ and\ \citenamefont
  {Brassard}(2014)}]{Bennett_QKD_2014}%
  \BibitemOpen
  \bibfield  {author} {\bibinfo {author} {\bibfnamefont {C.~H.}\ \bibnamefont
  {Bennett}}\ and\ \bibinfo {author} {\bibfnamefont {G.}~\bibnamefont
  {Brassard}},\ }\bibfield  {title} {\bibinfo {title} {Quantum cryptography:
  Public key distribution and coin tossing},\ }\href
  {https://doi.org/10.1016/j.tcs.2014.05.025} {\bibfield  {journal} {\bibinfo
  {journal} {Theor. Comput. Sci.}\ }\textbf {\bibinfo {volume} {560}},\
  \bibinfo {pages} {7} (\bibinfo {year} {2014})}\BibitemShut {NoStop}%
\bibitem [{\citenamefont {Cirac}\ \emph {et~al.}(1999)\citenamefont {Cirac},
  \citenamefont {Ekert}, \citenamefont {Huelga},\ and\ \citenamefont
  {Macchiavello}}]{Cirac.etal1999}%
  \BibitemOpen
  \bibfield  {author} {\bibinfo {author} {\bibfnamefont {J.~I.}\ \bibnamefont
  {Cirac}}, \bibinfo {author} {\bibfnamefont {A.~K.}\ \bibnamefont {Ekert}},
  \bibinfo {author} {\bibfnamefont {S.~F.}\ \bibnamefont {Huelga}},\ and\
  \bibinfo {author} {\bibfnamefont {C.}~\bibnamefont {Macchiavello}},\
  }\bibfield  {title} {\bibinfo {title} {Distributed quantum computation over
  noisy channels},\ }\href {https://doi.org/10.1103/PhysRevA.59.4249}
  {\bibfield  {journal} {\bibinfo  {journal} {Phys. Rev. A}\ }\textbf {\bibinfo
  {volume} {59}},\ \bibinfo {pages} {4249} (\bibinfo {year}
  {1999})}\BibitemShut {NoStop}%
\bibitem [{\citenamefont {Hayashi}\ and\ \citenamefont
  {Morimae}(2015)}]{Hayashi.Morimae2015}%
  \BibitemOpen
  \bibfield  {author} {\bibinfo {author} {\bibfnamefont {M.}~\bibnamefont
  {Hayashi}}\ and\ \bibinfo {author} {\bibfnamefont {T.}~\bibnamefont
  {Morimae}},\ }\bibfield  {title} {\bibinfo {title} {Verifiable
  measurement-only blind quantum computing with stabilizer testing},\ }\href
  {https://doi.org/10.1103/PhysRevLett.115.220502} {\bibfield  {journal}
  {\bibinfo  {journal} {Phys. Rev. Lett.}\ }\textbf {\bibinfo {volume} {115}},\
  \bibinfo {pages} {220502} (\bibinfo {year} {2015})}\BibitemShut {NoStop}%
\bibitem [{\citenamefont {Gheorghiu}\ \emph {et~al.}(2018)\citenamefont
  {Gheorghiu}, \citenamefont {Kapourniotis},\ and\ \citenamefont
  {Kashefi}}]{Gheorghiu_verification_2018}%
  \BibitemOpen
  \bibfield  {author} {\bibinfo {author} {\bibfnamefont {A.}~\bibnamefont
  {Gheorghiu}}, \bibinfo {author} {\bibfnamefont {T.}~\bibnamefont
  {Kapourniotis}},\ and\ \bibinfo {author} {\bibfnamefont {E.}~\bibnamefont
  {Kashefi}},\ }\bibfield  {title} {\bibinfo {title} {Verification of quantum
  computation: An overview of existing approaches},\ }\href
  {https://doi.org/10.1007/s00224-018-9872-3} {\bibfield  {journal} {\bibinfo
  {journal} {Theory Comput. Syst.}\ }\textbf {\bibinfo {volume} {63}},\
  \bibinfo {pages} {715} (\bibinfo {year} {2018})}\BibitemShut {NoStop}%
\bibitem [{\citenamefont {Paris}\ and\ \citenamefont
  {\v{R}eh\'a\v{c}ek}(2004)}]{QSE2004}%
  \BibitemOpen
  \bibinfo {editor} {\bibfnamefont {M.}~\bibnamefont {Paris}}\ and\ \bibinfo
  {editor} {\bibfnamefont {J.}~\bibnamefont {\v{R}eh\'a\v{c}ek}},\ eds.,\
  \href@noop {} {\emph {\bibinfo {title} {Quantum State Estimation}}},\
  \bibinfo {series} {Lecture Notes in Physics}, Vol.\ \bibinfo {volume} {649}\
  (\bibinfo  {publisher} {Springer-Verlag Berlin Heidelberg},\ \bibinfo {year}
  {2004})\BibitemShut {NoStop}%
\bibitem [{\citenamefont {Flammia}\ and\ \citenamefont
  {Liu}(2011)}]{Flammia.Liu2011}%
  \BibitemOpen
  \bibfield  {author} {\bibinfo {author} {\bibfnamefont {S.~T.}\ \bibnamefont
  {Flammia}}\ and\ \bibinfo {author} {\bibfnamefont {Y.-K.}\ \bibnamefont
  {Liu}},\ }\bibfield  {title} {\bibinfo {title} {Direct fidelity estimation
  from few {P}auli measurements},\ }\href
  {https://doi.org/10.1103/PhysRevLett.106.230501} {\bibfield  {journal}
  {\bibinfo  {journal} {Phys. Rev. Lett.}\ }\textbf {\bibinfo {volume} {106}},\
  \bibinfo {pages} {230501} (\bibinfo {year} {2011})}\BibitemShut {NoStop}%
\bibitem [{\citenamefont {Dimi{\'c}}\ and\ \citenamefont
  {Daki{\'c}}(2018)}]{Dimic.Dakic2018}%
  \BibitemOpen
  \bibfield  {author} {\bibinfo {author} {\bibfnamefont {A.}~\bibnamefont
  {Dimi{\'c}}}\ and\ \bibinfo {author} {\bibfnamefont {B.}~\bibnamefont
  {Daki{\'c}}},\ }\bibfield  {title} {\bibinfo {title} {Single-copy
  entanglement detection},\ }\href {https://doi.org/10.1038/s41534-017-0055-x}
  {\bibfield  {journal} {\bibinfo  {journal} {npj Quantum Inf.}\ }\textbf
  {\bibinfo {volume} {4}},\ \bibinfo {pages} {11} (\bibinfo {year}
  {2018})}\BibitemShut {NoStop}%
\bibitem [{\citenamefont {Pallister}\ \emph {et~al.}(2018)\citenamefont
  {Pallister}, \citenamefont {Linden},\ and\ \citenamefont
  {Montanaro}}]{Pallister.etal2018}%
  \BibitemOpen
  \bibfield  {author} {\bibinfo {author} {\bibfnamefont {S.}~\bibnamefont
  {Pallister}}, \bibinfo {author} {\bibfnamefont {N.}~\bibnamefont {Linden}},\
  and\ \bibinfo {author} {\bibfnamefont {A.}~\bibnamefont {Montanaro}},\
  }\bibfield  {title} {\bibinfo {title} {Optimal verification of entangled
  states with local measurements},\ }\href
  {https://doi.org/10.1103/PhysRevLett.120.170502} {\bibfield  {journal}
  {\bibinfo  {journal} {Phys. Rev. Lett.}\ }\textbf {\bibinfo {volume} {120}},\
  \bibinfo {pages} {170502} (\bibinfo {year} {2018})}\BibitemShut {NoStop}%
\bibitem [{\citenamefont {Morimae}\ \emph {et~al.}(2017)\citenamefont
  {Morimae}, \citenamefont {Takeuchi},\ and\ \citenamefont
  {Hayashi}}]{Morimae.etal2017}%
  \BibitemOpen
  \bibfield  {author} {\bibinfo {author} {\bibfnamefont {T.}~\bibnamefont
  {Morimae}}, \bibinfo {author} {\bibfnamefont {Y.}~\bibnamefont {Takeuchi}},\
  and\ \bibinfo {author} {\bibfnamefont {M.}~\bibnamefont {Hayashi}},\
  }\bibfield  {title} {\bibinfo {title} {Verification of hypergraph states},\
  }\href {https://doi.org/10.1103/PhysRevA.96.062321} {\bibfield  {journal}
  {\bibinfo  {journal} {Phys. Rev. A}\ }\textbf {\bibinfo {volume} {96}},\
  \bibinfo {pages} {062321} (\bibinfo {year} {2017})}\BibitemShut {NoStop}%
\bibitem [{\citenamefont {Takeuchi}\ and\ \citenamefont
  {Morimae}(2018)}]{Takeuchi.Morimae2018}%
  \BibitemOpen
  \bibfield  {author} {\bibinfo {author} {\bibfnamefont {Y.}~\bibnamefont
  {Takeuchi}}\ and\ \bibinfo {author} {\bibfnamefont {T.}~\bibnamefont
  {Morimae}},\ }\bibfield  {title} {\bibinfo {title} {Verification of
  many-qubit states},\ }\href {https://doi.org/10.1103/PhysRevX.8.021060}
  {\bibfield  {journal} {\bibinfo  {journal} {Phys. Rev. X}\ }\textbf {\bibinfo
  {volume} {8}},\ \bibinfo {pages} {021060} (\bibinfo {year}
  {2018})}\BibitemShut {NoStop}%
\bibitem [{\citenamefont {Yu}\ \emph {et~al.}(2019)\citenamefont {Yu},
  \citenamefont {Shang},\ and\ \citenamefont {G\"uhne}}]{Yu.etal2019}%
  \BibitemOpen
  \bibfield  {author} {\bibinfo {author} {\bibfnamefont {X.-D.}\ \bibnamefont
  {Yu}}, \bibinfo {author} {\bibfnamefont {J.}~\bibnamefont {Shang}},\ and\
  \bibinfo {author} {\bibfnamefont {O.}~\bibnamefont {G\"uhne}},\ }\bibfield
  {title} {\bibinfo {title} {Optimal verification of general bipartite pure
  states},\ }\href {https://doi.org/10.1038/s41534-019-0226-z} {\bibfield
  {journal} {\bibinfo  {journal} {npj Quantum Inf.}\ }\textbf {\bibinfo
  {volume} {5}},\ \bibinfo {pages} {112} (\bibinfo {year} {2019})}\BibitemShut
  {NoStop}%
\bibitem [{\citenamefont {Li}\ \emph {et~al.}(2019)\citenamefont {Li},
  \citenamefont {Han},\ and\ \citenamefont {Zhu}}]{Li.etal2019}%
  \BibitemOpen
  \bibfield  {author} {\bibinfo {author} {\bibfnamefont {Z.}~\bibnamefont
  {Li}}, \bibinfo {author} {\bibfnamefont {Y.-G.}\ \bibnamefont {Han}},\ and\
  \bibinfo {author} {\bibfnamefont {H.}~\bibnamefont {Zhu}},\ }\bibfield
  {title} {\bibinfo {title} {Efficient verification of bipartite pure states},\
  }\href {https://doi.org/10.1103/PhysRevA.100.032316} {\bibfield  {journal}
  {\bibinfo  {journal} {Phys. Rev. A}\ }\textbf {\bibinfo {volume} {100}},\
  \bibinfo {pages} {032316} (\bibinfo {year} {2019})}\BibitemShut {NoStop}%
\bibitem [{\citenamefont {Wang}\ and\ \citenamefont
  {Hayashi}(2019)}]{Wang.Hayashi2019}%
  \BibitemOpen
  \bibfield  {author} {\bibinfo {author} {\bibfnamefont {K.}~\bibnamefont
  {Wang}}\ and\ \bibinfo {author} {\bibfnamefont {M.}~\bibnamefont {Hayashi}},\
  }\bibfield  {title} {\bibinfo {title} {Optimal verification of two-qubit pure
  states},\ }\href {https://doi.org/10.1103/PhysRevA.100.032315} {\bibfield
  {journal} {\bibinfo  {journal} {Phys. Rev. A}\ }\textbf {\bibinfo {volume}
  {100}},\ \bibinfo {pages} {032315} (\bibinfo {year} {2019})}\BibitemShut
  {NoStop}%
\bibitem [{\citenamefont {{Zhu}}\ and\ \citenamefont
  {{Hayashi}}(2019)}]{Zhu.Hayashi2019a}%
  \BibitemOpen
  \bibfield  {author} {\bibinfo {author} {\bibfnamefont {H.}~\bibnamefont
  {{Zhu}}}\ and\ \bibinfo {author} {\bibfnamefont {M.}~\bibnamefont
  {{Hayashi}}},\ }\bibfield  {title} {\bibinfo {title} {Optimal verification
  and fidelity estimation of maximally entangled states},\ }\href
  {https://doi.org/10.1103/PhysRevA.99.052346} {\bibfield  {journal} {\bibinfo
  {journal} {Phys. Rev. A}\ }\textbf {\bibinfo {volume} {99}},\ \bibinfo
  {pages} {052346} (\bibinfo {year} {2019})}\BibitemShut {NoStop}%
\bibitem [{\citenamefont {Zhu}\ and\ \citenamefont
  {Hayashi}(2019{\natexlab{a}})}]{Zhu.Hayashi2019b}%
  \BibitemOpen
  \bibfield  {author} {\bibinfo {author} {\bibfnamefont {H.}~\bibnamefont
  {Zhu}}\ and\ \bibinfo {author} {\bibfnamefont {M.}~\bibnamefont {Hayashi}},\
  }\bibfield  {title} {\bibinfo {title} {Efficient verification of hypergraph
  states},\ }\href {https://doi.org/10.1103/PhysRevApplied.12.054047}
  {\bibfield  {journal} {\bibinfo  {journal} {Phys. Rev. Appl.}\ }\textbf
  {\bibinfo {volume} {12}},\ \bibinfo {pages} {054047} (\bibinfo {year}
  {2019}{\natexlab{a}})}\BibitemShut {NoStop}%
\bibitem [{\citenamefont {Zhu}\ and\ \citenamefont
  {Hayashi}(2019{\natexlab{b}})}]{Zhu.Hayashi2019c}%
  \BibitemOpen
  \bibfield  {author} {\bibinfo {author} {\bibfnamefont {H.}~\bibnamefont
  {Zhu}}\ and\ \bibinfo {author} {\bibfnamefont {M.}~\bibnamefont {Hayashi}},\
  }\bibfield  {title} {\bibinfo {title} {Efficient verification of pure quantum
  states in the adversarial scenario},\ }\href
  {https://doi.org/10.1103/PhysRevLett.123.260504} {\bibfield  {journal}
  {\bibinfo  {journal} {Phys. Rev. Lett.}\ }\textbf {\bibinfo {volume} {123}},\
  \bibinfo {pages} {260504} (\bibinfo {year} {2019}{\natexlab{b}})}\BibitemShut
  {NoStop}%
\bibitem [{\citenamefont {Zhu}\ and\ \citenamefont
  {Hayashi}(2019{\natexlab{c}})}]{Zhu.Hayashi2019d}%
  \BibitemOpen
  \bibfield  {author} {\bibinfo {author} {\bibfnamefont {H.}~\bibnamefont
  {Zhu}}\ and\ \bibinfo {author} {\bibfnamefont {M.}~\bibnamefont {Hayashi}},\
  }\bibfield  {title} {\bibinfo {title} {General framework for verifying pure
  quantum states in the adversarial scenario},\ }\href
  {https://doi.org/10.1103/PhysRevA.100.062335} {\bibfield  {journal} {\bibinfo
   {journal} {Phys. Rev. A}\ }\textbf {\bibinfo {volume} {100}},\ \bibinfo
  {pages} {062335} (\bibinfo {year} {2019}{\natexlab{c}})}\BibitemShut
  {NoStop}%
\bibitem [{\citenamefont {Liu}\ \emph {et~al.}(2019)\citenamefont {Liu},
  \citenamefont {Yu}, \citenamefont {Shang}, \citenamefont {Zhu},\ and\
  \citenamefont {Zhang}}]{Liu.etal2019b}%
  \BibitemOpen
  \bibfield  {author} {\bibinfo {author} {\bibfnamefont {Y.-C.}\ \bibnamefont
  {Liu}}, \bibinfo {author} {\bibfnamefont {X.-D.}\ \bibnamefont {Yu}},
  \bibinfo {author} {\bibfnamefont {J.}~\bibnamefont {Shang}}, \bibinfo
  {author} {\bibfnamefont {H.}~\bibnamefont {Zhu}},\ and\ \bibinfo {author}
  {\bibfnamefont {X.}~\bibnamefont {Zhang}},\ }\bibfield  {title} {\bibinfo
  {title} {Efficient verification of {D}icke states},\ }\href
  {https://doi.org/10.1103/PhysRevApplied.12.044020} {\bibfield  {journal}
  {\bibinfo  {journal} {Phys. Rev. Appl.}\ }\textbf {\bibinfo {volume} {12}},\
  \bibinfo {pages} {044020} (\bibinfo {year} {2019})}\BibitemShut {NoStop}%
\bibitem [{\citenamefont {Li}\ \emph {et~al.}(2020)\citenamefont {Li},
  \citenamefont {Han},\ and\ \citenamefont {Zhu}}]{Li.etal2020b}%
  \BibitemOpen
  \bibfield  {author} {\bibinfo {author} {\bibfnamefont {Z.}~\bibnamefont
  {Li}}, \bibinfo {author} {\bibfnamefont {Y.-G.}\ \bibnamefont {Han}},\ and\
  \bibinfo {author} {\bibfnamefont {H.}~\bibnamefont {Zhu}},\ }\bibfield
  {title} {\bibinfo {title} {Optimal verification of
  {G}reenberger-{H}orne-{Z}eilinger states},\ }\href
  {https://doi.org/10.1103/PhysRevApplied.13.054002} {\bibfield  {journal}
  {\bibinfo  {journal} {Phys. Rev. Appl.}\ }\textbf {\bibinfo {volume} {13}},\
  \bibinfo {pages} {054002} (\bibinfo {year} {2020})}\BibitemShut {NoStop}%
\bibitem [{\citenamefont {Dangniam}\ \emph {et~al.}(2020)\citenamefont
  {Dangniam}, \citenamefont {Han},\ and\ \citenamefont
  {Zhu}}]{Dangniam.etal2020}%
  \BibitemOpen
  \bibfield  {author} {\bibinfo {author} {\bibfnamefont {N.}~\bibnamefont
  {Dangniam}}, \bibinfo {author} {\bibfnamefont {Y.-G.}\ \bibnamefont {Han}},\
  and\ \bibinfo {author} {\bibfnamefont {H.}~\bibnamefont {Zhu}},\ }\bibfield
  {title} {\bibinfo {title} {Optimal verification of stabilizer states},\
  }\href {https://doi.org/10.1103/PhysRevResearch.2.043323} {\bibfield
  {journal} {\bibinfo  {journal} {Phys. Rev. Research}\ }\textbf {\bibinfo
  {volume} {2}},\ \bibinfo {pages} {043323} (\bibinfo {year}
  {2020})}\BibitemShut {NoStop}%
\bibitem [{\citenamefont {Zhang}\ \emph
  {et~al.}(2020{\natexlab{a}})\citenamefont {Zhang}, \citenamefont {Zhang},
  \citenamefont {Chen}, \citenamefont {Peng}, \citenamefont {Xu}, \citenamefont
  {Yin}, \citenamefont {Yu}, \citenamefont {Ye}, \citenamefont {Han},
  \citenamefont {Xu}, \citenamefont {Chen}, \citenamefont {Li},\ and\
  \citenamefont {Guo}}]{Zhang.etal2020a}%
  \BibitemOpen
  \bibfield  {author} {\bibinfo {author} {\bibfnamefont {W.-H.}\ \bibnamefont
  {Zhang}}, \bibinfo {author} {\bibfnamefont {C.}~\bibnamefont {Zhang}},
  \bibinfo {author} {\bibfnamefont {Z.}~\bibnamefont {Chen}}, \bibinfo {author}
  {\bibfnamefont {X.-X.}\ \bibnamefont {Peng}}, \bibinfo {author}
  {\bibfnamefont {X.-Y.}\ \bibnamefont {Xu}}, \bibinfo {author} {\bibfnamefont
  {P.}~\bibnamefont {Yin}}, \bibinfo {author} {\bibfnamefont {S.}~\bibnamefont
  {Yu}}, \bibinfo {author} {\bibfnamefont {X.-J.}\ \bibnamefont {Ye}}, \bibinfo
  {author} {\bibfnamefont {Y.-J.}\ \bibnamefont {Han}}, \bibinfo {author}
  {\bibfnamefont {J.-S.}\ \bibnamefont {Xu}}, \bibinfo {author} {\bibfnamefont
  {G.}~\bibnamefont {Chen}}, \bibinfo {author} {\bibfnamefont {C.-F.}\
  \bibnamefont {Li}},\ and\ \bibinfo {author} {\bibfnamefont {G.-C.}\
  \bibnamefont {Guo}},\ }\bibfield  {title} {\bibinfo {title} {Experimental
  optimal verification of entangled states using local measurements},\ }\href
  {https://doi.org/10.1103/PhysRevLett.125.030506} {\bibfield  {journal}
  {\bibinfo  {journal} {Phys. Rev. Lett.}\ }\textbf {\bibinfo {volume} {125}},\
  \bibinfo {pages} {030506} (\bibinfo {year} {2020}{\natexlab{a}})}\BibitemShut
  {NoStop}%
\bibitem [{\citenamefont {Jiang}\ \emph {et~al.}(2020)\citenamefont {Jiang},
  \citenamefont {Wang}, \citenamefont {Qian}, \citenamefont {Chen},
  \citenamefont {Chen}, \citenamefont {Lu}, \citenamefont {Xia}, \citenamefont
  {Song}, \citenamefont {Zhu},\ and\ \citenamefont {Ma}}]{Jiang.etal2020}%
  \BibitemOpen
  \bibfield  {author} {\bibinfo {author} {\bibfnamefont {X.}~\bibnamefont
  {Jiang}}, \bibinfo {author} {\bibfnamefont {K.}~\bibnamefont {Wang}},
  \bibinfo {author} {\bibfnamefont {K.}~\bibnamefont {Qian}}, \bibinfo {author}
  {\bibfnamefont {Z.}~\bibnamefont {Chen}}, \bibinfo {author} {\bibfnamefont
  {Z.}~\bibnamefont {Chen}}, \bibinfo {author} {\bibfnamefont {L.}~\bibnamefont
  {Lu}}, \bibinfo {author} {\bibfnamefont {L.}~\bibnamefont {Xia}}, \bibinfo
  {author} {\bibfnamefont {F.}~\bibnamefont {Song}}, \bibinfo {author}
  {\bibfnamefont {S.}~\bibnamefont {Zhu}},\ and\ \bibinfo {author}
  {\bibfnamefont {X.}~\bibnamefont {Ma}},\ }\bibfield  {title} {\bibinfo
  {title} {Towards the standardization of quantum state verification using
  optimal strategies},\ }\href {https://doi.org/10.1038/s41534-020-00317-7}
  {\bibfield  {journal} {\bibinfo  {journal} {npj Quantum Inf.}\ }\textbf
  {\bibinfo {volume} {6}},\ \bibinfo {pages} {90} (\bibinfo {year}
  {2020})}\BibitemShut {NoStop}%
\bibitem [{\citenamefont {Zhang}\ \emph
  {et~al.}(2020{\natexlab{b}})\citenamefont {Zhang}, \citenamefont {Liu},
  \citenamefont {Yin}, \citenamefont {Peng}, \citenamefont {Li}, \citenamefont
  {Xu}, \citenamefont {Yu}, \citenamefont {Hou}, \citenamefont {Han},
  \citenamefont {Xu}, \citenamefont {Zhou}, \citenamefont {Chen}, \citenamefont
  {Li},\ and\ \citenamefont {Guo}}]{Zhang.etal2020b}%
  \BibitemOpen
  \bibfield  {author} {\bibinfo {author} {\bibfnamefont {W.-H.}\ \bibnamefont
  {Zhang}}, \bibinfo {author} {\bibfnamefont {X.}~\bibnamefont {Liu}}, \bibinfo
  {author} {\bibfnamefont {P.}~\bibnamefont {Yin}}, \bibinfo {author}
  {\bibfnamefont {X.-X.}\ \bibnamefont {Peng}}, \bibinfo {author}
  {\bibfnamefont {G.-C.}\ \bibnamefont {Li}}, \bibinfo {author} {\bibfnamefont
  {X.-Y.}\ \bibnamefont {Xu}}, \bibinfo {author} {\bibfnamefont
  {S.}~\bibnamefont {Yu}}, \bibinfo {author} {\bibfnamefont {Z.-B.}\
  \bibnamefont {Hou}}, \bibinfo {author} {\bibfnamefont {Y.-J.}\ \bibnamefont
  {Han}}, \bibinfo {author} {\bibfnamefont {J.-S.}\ \bibnamefont {Xu}},
  \bibinfo {author} {\bibfnamefont {Z.-Q.}\ \bibnamefont {Zhou}}, \bibinfo
  {author} {\bibfnamefont {G.}~\bibnamefont {Chen}}, \bibinfo {author}
  {\bibfnamefont {C.-F.}\ \bibnamefont {Li}},\ and\ \bibinfo {author}
  {\bibfnamefont {G.-C.}\ \bibnamefont {Guo}},\ }\bibfield  {title} {\bibinfo
  {title} {Classical communication enhanced quantum state verification},\
  }\href {https://doi.org/10.1038/s41534-020-00328-4} {\bibfield  {journal}
  {\bibinfo  {journal} {npj Quantum Inf.}\ }\textbf {\bibinfo {volume} {6}},\
  \bibinfo {pages} {103} (\bibinfo {year} {2020}{\natexlab{b}})}\BibitemShut
  {NoStop}%
\bibitem [{\citenamefont {Li}\ \emph {et~al.}(2021{\natexlab{a}})\citenamefont
  {Li}, \citenamefont {Han}, \citenamefont {Sun}, \citenamefont {Shang},\ and\
  \citenamefont {Zhu}}]{Li.etal2020a}%
  \BibitemOpen
  \bibfield  {author} {\bibinfo {author} {\bibfnamefont {Z.}~\bibnamefont
  {Li}}, \bibinfo {author} {\bibfnamefont {Y.-G.}\ \bibnamefont {Han}},
  \bibinfo {author} {\bibfnamefont {H.-F.}\ \bibnamefont {Sun}}, \bibinfo
  {author} {\bibfnamefont {J.}~\bibnamefont {Shang}},\ and\ \bibinfo {author}
  {\bibfnamefont {H.}~\bibnamefont {Zhu}},\ }\bibfield  {title} {\bibinfo
  {title} {Verification of phased {D}icke states},\ }\href
  {https://doi.org/10.1103/PhysRevA.103.022601} {\bibfield  {journal} {\bibinfo
   {journal} {Phys. Rev. A}\ }\textbf {\bibinfo {volume} {103}},\ \bibinfo
  {pages} {022601} (\bibinfo {year} {2021}{\natexlab{a}})}\BibitemShut
  {NoStop}%
\bibitem [{\citenamefont {Liu}\ \emph {et~al.}(2021{\natexlab{a}})\citenamefont
  {Liu}, \citenamefont {Shang}, \citenamefont {Han},\ and\ \citenamefont
  {Zhang}}]{Liu.etal2020b}%
  \BibitemOpen
  \bibfield  {author} {\bibinfo {author} {\bibfnamefont {Y.-C.}\ \bibnamefont
  {Liu}}, \bibinfo {author} {\bibfnamefont {J.}~\bibnamefont {Shang}}, \bibinfo
  {author} {\bibfnamefont {R.}~\bibnamefont {Han}},\ and\ \bibinfo {author}
  {\bibfnamefont {X.}~\bibnamefont {Zhang}},\ }\bibfield  {title} {\bibinfo
  {title} {Universally optimal verification of entangled states with
  nondemolition measurements},\ }\href
  {https://doi.org/10.1103/PhysRevLett.126.090504} {\bibfield  {journal}
  {\bibinfo  {journal} {Phys. Rev. Lett.}\ }\textbf {\bibinfo {volume} {126}},\
  \bibinfo {pages} {090504} (\bibinfo {year} {2021}{\natexlab{a}})}\BibitemShut
  {NoStop}%
\bibitem [{\citenamefont {Liu}\ \emph {et~al.}(2021{\natexlab{b}})\citenamefont
  {Liu}, \citenamefont {Shang},\ and\ \citenamefont {Zhang}}]{Liu.etal2021}%
  \BibitemOpen
  \bibfield  {author} {\bibinfo {author} {\bibfnamefont {Y.-C.}\ \bibnamefont
  {Liu}}, \bibinfo {author} {\bibfnamefont {J.}~\bibnamefont {Shang}},\ and\
  \bibinfo {author} {\bibfnamefont {X.}~\bibnamefont {Zhang}},\ }\bibfield
  {title} {\bibinfo {title} {Efficient verification of entangled
  continuous-variable quantum states with local measurements},\ }\href
  {https://doi.org/10.1103/PhysRevResearch.3.L042004} {\bibfield  {journal}
  {\bibinfo  {journal} {Phys. Rev. Research}\ }\textbf {\bibinfo {volume}
  {3}},\ \bibinfo {pages} {L042004} (\bibinfo {year}
  {2021}{\natexlab{b}})}\BibitemShut {NoStop}%
\bibitem [{\citenamefont {Han}\ \emph {et~al.}(2021)\citenamefont {Han},
  \citenamefont {Li}, \citenamefont {Wang},\ and\ \citenamefont
  {Zhu}}]{Han.etal2021}%
  \BibitemOpen
  \bibfield  {author} {\bibinfo {author} {\bibfnamefont {Y.-G.}\ \bibnamefont
  {Han}}, \bibinfo {author} {\bibfnamefont {Z.}~\bibnamefont {Li}}, \bibinfo
  {author} {\bibfnamefont {Y.}~\bibnamefont {Wang}},\ and\ \bibinfo {author}
  {\bibfnamefont {H.}~\bibnamefont {Zhu}},\ }\bibfield  {title} {\bibinfo
  {title} {Optimal verification of the {B}ell state and
  {G}reenberger--{H}orne--{Z}eilinger states in untrusted quantum networks},\
  }\href {https://doi.org/https://doi.org/10.1038/s41534-021-00499-8}
  {\bibfield  {journal} {\bibinfo  {journal} {npj Quantum Inf.}\ }\textbf
  {\bibinfo {volume} {7}},\ \bibinfo {pages} {164} (\bibinfo {year}
  {2021})}\BibitemShut {NoStop}%
\bibitem [{\citenamefont {Zhu}\ \emph {et~al.}(2024)\citenamefont {Zhu},
  \citenamefont {Li},\ and\ \citenamefont {Chen}}]{Zhu.etal2022}%
  \BibitemOpen
  \bibfield  {author} {\bibinfo {author} {\bibfnamefont {H.}~\bibnamefont
  {Zhu}}, \bibinfo {author} {\bibfnamefont {Y.}~\bibnamefont {Li}},\ and\
  \bibinfo {author} {\bibfnamefont {T.}~\bibnamefont {Chen}},\ }\bibfield
  {title} {\bibinfo {title} {Efficient verification of ground states of
  frustration-free hamiltonians},\ }\href
  {https://doi.org/10.22331/q-2024-01-10-1221} {\bibfield  {journal} {\bibinfo
  {journal} {Quantum}\ }\textbf {\bibinfo {volume} {8}},\ \bibinfo {pages}
  {1221} (\bibinfo {year} {2024})}\BibitemShut {NoStop}%
\bibitem [{\citenamefont {Chen}\ \emph {et~al.}(2023)\citenamefont {Chen},
  \citenamefont {Li},\ and\ \citenamefont {Zhu}}]{ChenHuang.etal2023}%
  \BibitemOpen
  \bibfield  {author} {\bibinfo {author} {\bibfnamefont {T.}~\bibnamefont
  {Chen}}, \bibinfo {author} {\bibfnamefont {Y.}~\bibnamefont {Li}},\ and\
  \bibinfo {author} {\bibfnamefont {H.}~\bibnamefont {Zhu}},\ }\bibfield
  {title} {\bibinfo {title} {Efficient verification of
  {A}ffleck-{K}ennedy-{L}ieb-{T}asaki states},\ }\href
  {https://doi.org/10.1103/PhysRevA.107.022616} {\bibfield  {journal} {\bibinfo
   {journal} {Phys. Rev. A}\ }\textbf {\bibinfo {volume} {107}},\ \bibinfo
  {pages} {022616} (\bibinfo {year} {2023})}\BibitemShut {NoStop}%
\bibitem [{\citenamefont {Li}\ \emph {et~al.}(2021{\natexlab{b}})\citenamefont
  {Li}, \citenamefont {Zhang}, \citenamefont {Li},\ and\ \citenamefont
  {Zhu}}]{Li.Y.etal2021}%
  \BibitemOpen
  \bibfield  {author} {\bibinfo {author} {\bibfnamefont {Y.}~\bibnamefont
  {Li}}, \bibinfo {author} {\bibfnamefont {H.}~\bibnamefont {Zhang}}, \bibinfo
  {author} {\bibfnamefont {Z.}~\bibnamefont {Li}},\ and\ \bibinfo {author}
  {\bibfnamefont {H.}~\bibnamefont {Zhu}},\ }\bibfield  {title} {\bibinfo
  {title} {Minimum number of experimental settings required to verify bipartite
  pure states and unitaries},\ }\href
  {https://doi.org/10.1103/PhysRevA.104.062439} {\bibfield  {journal} {\bibinfo
   {journal} {Phys. Rev. A}\ }\textbf {\bibinfo {volume} {104}},\ \bibinfo
  {pages} {062439} (\bibinfo {year} {2021}{\natexlab{b}})}\BibitemShut
  {NoStop}%
\bibitem [{\citenamefont {Yu}\ \emph {et~al.}(2022)\citenamefont {Yu},
  \citenamefont {Shang},\ and\ \citenamefont {Gühne}}]{Yu.etal2022}%
  \BibitemOpen
  \bibfield  {author} {\bibinfo {author} {\bibfnamefont {X.-D.}\ \bibnamefont
  {Yu}}, \bibinfo {author} {\bibfnamefont {J.}~\bibnamefont {Shang}},\ and\
  \bibinfo {author} {\bibfnamefont {O.}~\bibnamefont {Gühne}},\ }\bibfield
  {title} {\bibinfo {title} {Statistical methods for quantum state verification
  and fidelity estimation},\ }\href
  {https://doi.org/https://doi.org/10.1002/qute.202100126} {\bibfield
  {journal} {\bibinfo  {journal} {Adv. Quantum Technol.}\ }\textbf {\bibinfo
  {volume} {5}},\ \bibinfo {pages} {2100126} (\bibinfo {year}
  {2022})}\BibitemShut {NoStop}%
\bibitem [{\citenamefont {Miguel-Ramiro}\ \emph {et~al.}(2022)\citenamefont
  {Miguel-Ramiro}, \citenamefont {Riera-Sàbat},\ and\ \citenamefont
  {Dür}}]{miguel-ramiro_collective_2022}%
  \BibitemOpen
  \bibfield  {author} {\bibinfo {author} {\bibfnamefont {J.}~\bibnamefont
  {Miguel-Ramiro}}, \bibinfo {author} {\bibfnamefont {F.}~\bibnamefont
  {Riera-Sàbat}},\ and\ \bibinfo {author} {\bibfnamefont {W.}~\bibnamefont
  {Dür}},\ }\bibfield  {title} {\bibinfo {title} {Collective operations can
  exponentially enhance quantum state verification},\ }\href
  {https://doi.org/10.1103/PhysRevLett.129.190504} {\bibfield  {journal}
  {\bibinfo  {journal} {Phys. Rev. Lett.}\ }\textbf {\bibinfo {volume} {129}},\
  \bibinfo {pages} {190504} (\bibinfo {year} {2022})}\BibitemShut {NoStop}%
\bibitem [{\citenamefont {Chen}\ \emph {et~al.}(2025)\citenamefont {Chen},
  \citenamefont {Xie}, \citenamefont {Xu},\ and\ \citenamefont
  {Wang}}]{chen_memoryQSV_2025}%
  \BibitemOpen
  \bibfield  {author} {\bibinfo {author} {\bibfnamefont {S.}~\bibnamefont
  {Chen}}, \bibinfo {author} {\bibfnamefont {W.}~\bibnamefont {Xie}}, \bibinfo
  {author} {\bibfnamefont {P.}~\bibnamefont {Xu}},\ and\ \bibinfo {author}
  {\bibfnamefont {K.}~\bibnamefont {Wang}},\ }\bibfield  {title} {\bibinfo
  {title} {Quantum memory assisted entangled state verification with local
  measurements},\ }\href {https://doi.org/10.1103/PhysRevResearch.7.013003}
  {\bibfield  {journal} {\bibinfo  {journal} {Phys. Rev. Res.}\ }\textbf
  {\bibinfo {volume} {7}},\ \bibinfo {pages} {013003} (\bibinfo {year}
  {2025})}\BibitemShut {NoStop}%
\bibitem [{\citenamefont {Gross}\ \emph {et~al.}(2010)\citenamefont {Gross},
  \citenamefont {Liu}, \citenamefont {Flammia}, \citenamefont {Becker},\ and\
  \citenamefont {Eisert}}]{Gross.etal2010}%
  \BibitemOpen
  \bibfield  {author} {\bibinfo {author} {\bibfnamefont {D.}~\bibnamefont
  {Gross}}, \bibinfo {author} {\bibfnamefont {Y.-K.}\ \bibnamefont {Liu}},
  \bibinfo {author} {\bibfnamefont {S.~T.}\ \bibnamefont {Flammia}}, \bibinfo
  {author} {\bibfnamefont {S.}~\bibnamefont {Becker}},\ and\ \bibinfo {author}
  {\bibfnamefont {J.}~\bibnamefont {Eisert}},\ }\bibfield  {title} {\bibinfo
  {title} {Quantum state tomography via compressed sensing},\ }\href
  {https://doi.org/10.1103/PhysRevLett.105.150401} {\bibfield  {journal}
  {\bibinfo  {journal} {Phys. Rev. Lett.}\ }\textbf {\bibinfo {volume} {105}},\
  \bibinfo {pages} {150401} (\bibinfo {year} {2010})}\BibitemShut {NoStop}%
\bibitem [{\citenamefont {Gross}(2011)}]{Gross2011}%
  \BibitemOpen
  \bibfield  {author} {\bibinfo {author} {\bibfnamefont {D.}~\bibnamefont
  {Gross}},\ }\bibfield  {title} {\bibinfo {title} {Recovering low-rank
  matrices from few coefficients in any basis},\ }\href
  {https://doi.org/10.1109/TIT.2011.2104999} {\bibfield  {journal} {\bibinfo
  {journal} {IEEE Trans. Inf. Theory}\ }\textbf {\bibinfo {volume} {57}},\
  \bibinfo {pages} {1548} (\bibinfo {year} {2011})}\BibitemShut {NoStop}%
\bibitem [{\citenamefont {Ma}\ and\ \citenamefont
  {Shang}(2025)}]{ma_corruptedsensing_2025}%
  \BibitemOpen
  \bibfield  {author} {\bibinfo {author} {\bibfnamefont {M.}~\bibnamefont
  {Ma}}\ and\ \bibinfo {author} {\bibfnamefont {J.}~\bibnamefont {Shang}},\
  }\bibfield  {title} {\bibinfo {title} {Corrupted sensing quantum state
  tomography},\ }\href {https://doi.org/10.1088/1367-2630/adcfbf} {\bibfield
  {journal} {\bibinfo  {journal} {New J. Phys.}\ }\textbf {\bibinfo {volume}
  {27}},\ \bibinfo {pages} {054501} (\bibinfo {year} {2025})}\BibitemShut
  {NoStop}%
\bibitem [{\citenamefont {Flammia}\ and\ \citenamefont
  {Wallman}(2020)}]{flammia_efficient_2020}%
  \BibitemOpen
  \bibfield  {author} {\bibinfo {author} {\bibfnamefont {S.~T.}\ \bibnamefont
  {Flammia}}\ and\ \bibinfo {author} {\bibfnamefont {J.~J.}\ \bibnamefont
  {Wallman}},\ }\bibfield  {title} {\bibinfo {title} {Efficient estimation of
  {Pauli} channels},\ }\href {https://doi.org/10.1145/3408039} {\bibfield
  {journal} {\bibinfo  {journal} {ACM Trans. Quantum Comput.}\ }\textbf
  {\bibinfo {volume} {1}},\ \bibinfo {pages} {1} (\bibinfo {year}
  {2020})}\BibitemShut {NoStop}%
\bibitem [{\citenamefont {Harper}\ \emph {et~al.}(2020)\citenamefont {Harper},
  \citenamefont {Flammia},\ and\ \citenamefont
  {Wallman}}]{harper_efficient_2020}%
  \BibitemOpen
  \bibfield  {author} {\bibinfo {author} {\bibfnamefont {R.}~\bibnamefont
  {Harper}}, \bibinfo {author} {\bibfnamefont {S.~T.}\ \bibnamefont
  {Flammia}},\ and\ \bibinfo {author} {\bibfnamefont {J.~J.}\ \bibnamefont
  {Wallman}},\ }\bibfield  {title} {\bibinfo {title} {Efficient learning of
  quantum noise},\ }\href {https://doi.org/10.1038/s41567-020-0992-8}
  {\bibfield  {journal} {\bibinfo  {journal} {Nat. Phys.}\ }\textbf {\bibinfo
  {volume} {16}},\ \bibinfo {pages} {1184} (\bibinfo {year}
  {2020})}\BibitemShut {NoStop}%
\bibitem [{\citenamefont {Harper}\ \emph {et~al.}(2021)\citenamefont {Harper},
  \citenamefont {Yu},\ and\ \citenamefont {Flammia}}]{harper_fast_2021}%
  \BibitemOpen
  \bibfield  {author} {\bibinfo {author} {\bibfnamefont {R.}~\bibnamefont
  {Harper}}, \bibinfo {author} {\bibfnamefont {W.}~\bibnamefont {Yu}},\ and\
  \bibinfo {author} {\bibfnamefont {S.~T.}\ \bibnamefont {Flammia}},\
  }\bibfield  {title} {\bibinfo {title} {Fast estimation of sparse quantum
  noise},\ }\href {https://doi.org/10.1103/PRXQuantum.2.010322} {\bibfield
  {journal} {\bibinfo  {journal} {PRX Quantum}\ }\textbf {\bibinfo {volume}
  {2}},\ \bibinfo {pages} {010322} (\bibinfo {year} {2021})}\BibitemShut
  {NoStop}%
\bibitem [{\citenamefont {Barenco}\ \emph {et~al.}()\citenamefont {Barenco},
  \citenamefont {Berthiaume}, \citenamefont {Deutsch}, \citenamefont {Ekert},
  \citenamefont {Jozsa},\ and\ \citenamefont
  {Macchiavello}}]{barenco_stabilisation_1996}%
  \BibitemOpen
  \bibfield  {author} {\bibinfo {author} {\bibfnamefont {A.}~\bibnamefont
  {Barenco}}, \bibinfo {author} {\bibfnamefont {A.}~\bibnamefont {Berthiaume}},
  \bibinfo {author} {\bibfnamefont {D.}~\bibnamefont {Deutsch}}, \bibinfo
  {author} {\bibfnamefont {A.}~\bibnamefont {Ekert}}, \bibinfo {author}
  {\bibfnamefont {R.}~\bibnamefont {Jozsa}},\ and\ \bibinfo {author}
  {\bibfnamefont {C.}~\bibnamefont {Macchiavello}},\ }\bibfield  {title}
  {\bibinfo {title} {Stabilisation of quantum computations by symmetrisation},\ }\href
  {https://doi.org/10.1137/S0097539796302452} {\bibfield  {journal}
  {\bibinfo  {journal} {SIAM J. Comput.}\ }\textbf {\bibinfo {volume} {26}},\
  \bibinfo {pages} {1541} (\bibinfo {year} {1997})}\BibitemShut {NoStop}%
\bibitem [{\citenamefont {Buhrman}\ \emph {et~al.}(2001)\citenamefont
  {Buhrman}, \citenamefont {Cleve}, \citenamefont {Watrous},\ and\
  \citenamefont {de~Wolf}}]{buhrman_quantum_2001}%
  \BibitemOpen
  \bibfield  {author} {\bibinfo {author} {\bibfnamefont {H.}~\bibnamefont
  {Buhrman}}, \bibinfo {author} {\bibfnamefont {R.}~\bibnamefont {Cleve}},
  \bibinfo {author} {\bibfnamefont {J.}~\bibnamefont {Watrous}},\ and\ \bibinfo
  {author} {\bibfnamefont {R.}~\bibnamefont {de~Wolf}},\ }\bibfield  {title}
  {\bibinfo {title} {Quantum fingerprinting},\ }\href
  {https://doi.org/10.1103/PhysRevLett.87.167902} {\bibfield  {journal}
  {\bibinfo  {journal} {Phys. Rev. Lett.}\ }\textbf {\bibinfo {volume} {87}},\
  \bibinfo {pages} {167902} (\bibinfo {year} {2001})}\BibitemShut {NoStop}%
\bibitem [{\citenamefont {Ricci}\ \emph {et~al.}(2004)\citenamefont {Ricci},
  \citenamefont {De~Martini}, \citenamefont {Cerf}, \citenamefont {Filip},
  \citenamefont {Fiurasek},\ and\ \citenamefont
  {Macchiavello}}]{ricci_experimental_2004}%
  \BibitemOpen
  \bibfield  {author} {\bibinfo {author} {\bibfnamefont {M.}~\bibnamefont
  {Ricci}}, \bibinfo {author} {\bibfnamefont {F.}~\bibnamefont {De~Martini}},
  \bibinfo {author} {\bibfnamefont {N.~J.}\ \bibnamefont {Cerf}}, \bibinfo
  {author} {\bibfnamefont {R.}~\bibnamefont {Filip}}, \bibinfo {author}
  {\bibfnamefont {J.}~\bibnamefont {Fiurasek}},\ and\ \bibinfo {author}
  {\bibfnamefont {C.}~\bibnamefont {Macchiavello}},\ }\bibfield  {title}
  {\bibinfo {title} {Experimental purification of single qubits},\ }\href
  {https://doi.org/10.1103/PhysRevLett.93.170501} {\bibfield  {journal}
  {\bibinfo  {journal} {Phys. Rev. Lett.}\ }\textbf {\bibinfo {volume} {93}},\
  \bibinfo {pages} {170501} (\bibinfo {year} {2004})}\BibitemShut {NoStop}%
\bibitem [{\citenamefont {Childs}\ \emph {et~al.}(2025)\citenamefont {Childs},
  \citenamefont {Fu}, \citenamefont {Leung}, \citenamefont {Li}, \citenamefont
  {Ozols},\ and\ \citenamefont {Vyas}}]{Childs_2025_streaming}%
  \BibitemOpen
  \bibfield  {author} {\bibinfo {author} {\bibfnamefont {A.~M.}\ \bibnamefont
  {Childs}}, \bibinfo {author} {\bibfnamefont {H.}~\bibnamefont {Fu}}, \bibinfo
  {author} {\bibfnamefont {D.}~\bibnamefont {Leung}}, \bibinfo {author}
  {\bibfnamefont {Z.}~\bibnamefont {Li}}, \bibinfo {author} {\bibfnamefont
  {M.}~\bibnamefont {Ozols}},\ and\ \bibinfo {author} {\bibfnamefont
  {V.}~\bibnamefont {Vyas}},\ }\bibfield  {title} {\bibinfo {title} {Streaming
  quantum state purification},\ }\href
  {https://doi.org/10.22331/q-2025-01-21-1603} {\bibfield  {journal} {\bibinfo
  {journal} {Quantum}\ }\textbf {\bibinfo {volume} {9}},\ \bibinfo {pages}
  {1603} (\bibinfo {year} {2025})}\BibitemShut {NoStop}%
\bibitem [{\citenamefont {Cotler}\ \emph {et~al.}(2019)\citenamefont {Cotler},
  \citenamefont {Choi}, \citenamefont {Lukin}, \citenamefont {Gharibyan},
  \citenamefont {Grover}, \citenamefont {Tai}, \citenamefont {Rispoli},
  \citenamefont {Schittko}, \citenamefont {Preiss}, \citenamefont {Kaufman},
  \citenamefont {Greiner}, \citenamefont {Pichler},\ and\ \citenamefont
  {Hayden}}]{cotler_quantum_2019}%
  \BibitemOpen
  \bibfield  {author} {\bibinfo {author} {\bibfnamefont {J.}~\bibnamefont
  {Cotler}}, \bibinfo {author} {\bibfnamefont {S.}~\bibnamefont {Choi}},
  \bibinfo {author} {\bibfnamefont {A.}~\bibnamefont {Lukin}}, \bibinfo
  {author} {\bibfnamefont {H.}~\bibnamefont {Gharibyan}}, \bibinfo {author}
  {\bibfnamefont {T.}~\bibnamefont {Grover}}, \bibinfo {author} {\bibfnamefont
  {M.~E.}\ \bibnamefont {Tai}}, \bibinfo {author} {\bibfnamefont
  {M.}~\bibnamefont {Rispoli}}, \bibinfo {author} {\bibfnamefont
  {R.}~\bibnamefont {Schittko}}, \bibinfo {author} {\bibfnamefont {P.~M.}\
  \bibnamefont {Preiss}}, \bibinfo {author} {\bibfnamefont {A.~M.}\
  \bibnamefont {Kaufman}}, \bibinfo {author} {\bibfnamefont {M.}~\bibnamefont
  {Greiner}}, \bibinfo {author} {\bibfnamefont {H.}~\bibnamefont {Pichler}},\
  and\ \bibinfo {author} {\bibfnamefont {P.}~\bibnamefont {Hayden}},\
  }\bibfield  {title} {\bibinfo {title} {Quantum virtual cooling},\ }\href
  {https://doi.org/10.1103/PhysRevX.9.031013} {\bibfield  {journal} {\bibinfo
  {journal} {Phys. Rev. X}\ }\textbf {\bibinfo {volume} {9}},\ \bibinfo {pages}
  {031013} (\bibinfo {year} {2019})}\BibitemShut {NoStop}%
\bibitem [{\citenamefont {Huggins}\ \emph {et~al.}(2021)\citenamefont
  {Huggins}, \citenamefont {McArdle}, \citenamefont {O'Brien}, \citenamefont
  {Lee}, \citenamefont {Rubin}, \citenamefont {Boixo}, \citenamefont {Whaley},
  \citenamefont {Babbush},\ and\ \citenamefont
  {McClean}}]{Huggins_2021_virtual}%
  \BibitemOpen
  \bibfield  {author} {\bibinfo {author} {\bibfnamefont {W.~J.}\ \bibnamefont
  {Huggins}}, \bibinfo {author} {\bibfnamefont {S.}~\bibnamefont {McArdle}},
  \bibinfo {author} {\bibfnamefont {T.~E.}\ \bibnamefont {O'Brien}}, \bibinfo
  {author} {\bibfnamefont {J.}~\bibnamefont {Lee}}, \bibinfo {author}
  {\bibfnamefont {N.~C.}\ \bibnamefont {Rubin}}, \bibinfo {author}
  {\bibfnamefont {S.}~\bibnamefont {Boixo}}, \bibinfo {author} {\bibfnamefont
  {K.~B.}\ \bibnamefont {Whaley}}, \bibinfo {author} {\bibfnamefont
  {R.}~\bibnamefont {Babbush}},\ and\ \bibinfo {author} {\bibfnamefont {J.~R.}\
  \bibnamefont {McClean}},\ }\bibfield  {title} {\bibinfo {title} {Virtual
  distillation for quantum error mitigation},\ }\href
  {https://doi.org/10.1103/PhysRevX.11.041036} {\bibfield  {journal} {\bibinfo
  {journal} {Phys. Rev. X}\ }\textbf {\bibinfo {volume} {11}},\ \bibinfo
  {pages} {041036} (\bibinfo {year} {2021})}\BibitemShut {NoStop}%
\bibitem [{\citenamefont {Koczor}(2021)}]{Koczor_2021_exponential}%
  \BibitemOpen
  \bibfield  {author} {\bibinfo {author} {\bibfnamefont {B.}~\bibnamefont
  {Koczor}},\ }\bibfield  {title} {\bibinfo {title} {Exponential error
  suppression for near-term quantum devices},\ }\href
  {https://doi.org/10.1103/PhysRevX.11.031057} {\bibfield  {journal} {\bibinfo
  {journal} {Phys. Rev. X}\ }\textbf {\bibinfo {volume} {11}},\ \bibinfo
  {pages} {031057} (\bibinfo {year} {2021})}\BibitemShut {NoStop}%
\bibitem [{\citenamefont {Czarnik}\ \emph {et~al.}()\citenamefont {Czarnik},
  \citenamefont {Arrasmith}, \citenamefont {Cincio},\ and\ \citenamefont
  {Coles}}]{czarnik_2021_qubit}%
  \BibitemOpen
  \bibfield  {author} {\bibinfo {author} {\bibfnamefont {P.}~\bibnamefont
  {Czarnik}}, \bibinfo {author} {\bibfnamefont {A.}~\bibnamefont {Arrasmith}},
  \bibinfo {author} {\bibfnamefont {L.}~\bibnamefont {Cincio}},\ and\ \bibinfo
  {author} {\bibfnamefont {P.~J.}\ \bibnamefont {Coles}},\ }\href@noop {}
  {\bibinfo {title} {Qubit-efficient exponential suppression of errors}},\
  \Eprint {https://arxiv.org/abs/2102.06056} {arXiv:2102.06056} \BibitemShut
  {NoStop}%
\bibitem [{\citenamefont {Quek}\ \emph {et~al.}(2024)\citenamefont {Quek},
  \citenamefont {Kaur},\ and\ \citenamefont {Wilde}}]{Quek_2024_multivariate}%
  \BibitemOpen
  \bibfield  {author} {\bibinfo {author} {\bibfnamefont {Y.}~\bibnamefont
  {Quek}}, \bibinfo {author} {\bibfnamefont {E.}~\bibnamefont {Kaur}},\ and\
  \bibinfo {author} {\bibfnamefont {M.~M.}\ \bibnamefont {Wilde}},\ }\bibfield
  {title} {\bibinfo {title} {Multivariate trace estimation in constant quantum
  depth},\ }\href {https://doi.org/10.22331/q-2024-01-10-1220} {\bibfield
  {journal} {\bibinfo  {journal} {Quantum}\ }\textbf {\bibinfo {volume} {8}},\
  \bibinfo {pages} {1220} (\bibinfo {year} {2024})}\BibitemShut {NoStop}%
\bibitem [{\citenamefont {Ekert}\ \emph {et~al.}(2002)\citenamefont {Ekert},
  \citenamefont {Alves}, \citenamefont {Oi}, \citenamefont {Horodecki},
  \citenamefont {Horodecki},\ and\ \citenamefont
  {Kwek}}]{Ekert_estimation_2002}%
  \BibitemOpen
  \bibfield  {author} {\bibinfo {author} {\bibfnamefont {A.~K.}\ \bibnamefont
  {Ekert}}, \bibinfo {author} {\bibfnamefont {C.~M.}\ \bibnamefont {Alves}},
  \bibinfo {author} {\bibfnamefont {D.~K.~L.}\ \bibnamefont {Oi}}, \bibinfo
  {author} {\bibfnamefont {M.}~\bibnamefont {Horodecki}}, \bibinfo {author}
  {\bibfnamefont {P.}~\bibnamefont {Horodecki}},\ and\ \bibinfo {author}
  {\bibfnamefont {L.~C.}\ \bibnamefont {Kwek}},\ }\bibfield  {title} {\bibinfo
  {title} {Direct estimations of linear and nonlinear functionals of a quantum
  state},\ }\href {https://doi.org/10.1103/PhysRevLett.88.217901} {\bibfield
  {journal} {\bibinfo  {journal} {Phys. Rev. Lett.}\ }\textbf {\bibinfo
  {volume} {88}},\ \bibinfo {pages} {217901} (\bibinfo {year}
  {2002})}\BibitemShut {NoStop}%
\bibitem [{\citenamefont {Patel}\ \emph {et~al.}(2016)\citenamefont {Patel},
  \citenamefont {Ho}, \citenamefont {Ferreyrol}, \citenamefont {Ralph},\ and\
  \citenamefont {Pryde}}]{patel_quantum_2016}%
  \BibitemOpen
  \bibfield  {author} {\bibinfo {author} {\bibfnamefont {R.~B.}\ \bibnamefont
  {Patel}}, \bibinfo {author} {\bibfnamefont {J.}~\bibnamefont {Ho}}, \bibinfo
  {author} {\bibfnamefont {F.}~\bibnamefont {Ferreyrol}}, \bibinfo {author}
  {\bibfnamefont {T.~C.}\ \bibnamefont {Ralph}},\ and\ \bibinfo {author}
  {\bibfnamefont {G.~J.}\ \bibnamefont {Pryde}},\ }\bibfield  {title} {\bibinfo
  {title} {A quantum {Fredkin} gate},\ }\href
  {https://doi.org/10.1126/sciadv.1501531} {\bibfield  {journal} {\bibinfo
  {journal} {Sci. Adv.}\ }\textbf {\bibinfo {volume} {2}},\ \bibinfo {pages}
  {e1501531} (\bibinfo {year} {2016})}\BibitemShut {NoStop}%
\bibitem [{\citenamefont {Ono}\ \emph {et~al.}(2017)\citenamefont {Ono},
  \citenamefont {Okamoto}, \citenamefont {Tanida}, \citenamefont {Hofmann},\
  and\ \citenamefont {Takeuchi}}]{ono_implementation_2017}%
  \BibitemOpen
  \bibfield  {author} {\bibinfo {author} {\bibfnamefont {T.}~\bibnamefont
  {Ono}}, \bibinfo {author} {\bibfnamefont {R.}~\bibnamefont {Okamoto}},
  \bibinfo {author} {\bibfnamefont {M.}~\bibnamefont {Tanida}}, \bibinfo
  {author} {\bibfnamefont {H.~F.}\ \bibnamefont {Hofmann}},\ and\ \bibinfo
  {author} {\bibfnamefont {S.}~\bibnamefont {Takeuchi}},\ }\bibfield  {title}
  {\bibinfo {title} {Implementation of a quantum controlled-{SWAP} gate with
  photonic circuits},\ }\href {https://doi.org/10.1038/srep45353} {\bibfield
  {journal} {\bibinfo  {journal} {Sci. Rep.}\ }\textbf {\bibinfo {volume}
  {7}},\ \bibinfo {pages} {45353} (\bibinfo {year} {2017})}\BibitemShut
  {NoStop}%
\bibitem [{\citenamefont {St{\'a}rek}\ \emph {et~al.}(2018)\citenamefont
  {St{\'a}rek}, \citenamefont {Mi{\v{c}}uda}, \citenamefont {Mikov{\'a}},
  \citenamefont {Straka}, \citenamefont {Du{\v{s}}ek}, \citenamefont {Marek},
  \citenamefont {Je{\v{z}}ek}, \citenamefont {Filip},\ and\ \citenamefont
  {Fiur{\'a}{\v{s}}ek}}]{starek_nondestructive_2018}%
  \BibitemOpen
  \bibfield  {author} {\bibinfo {author} {\bibfnamefont {R.}~\bibnamefont
  {St{\'a}rek}}, \bibinfo {author} {\bibfnamefont {M.}~\bibnamefont
  {Mi{\v{c}}uda}}, \bibinfo {author} {\bibfnamefont {M.}~\bibnamefont
  {Mikov{\'a}}}, \bibinfo {author} {\bibfnamefont {I.}~\bibnamefont {Straka}},
  \bibinfo {author} {\bibfnamefont {M.}~\bibnamefont {Du{\v{s}}ek}}, \bibinfo
  {author} {\bibfnamefont {P.}~\bibnamefont {Marek}}, \bibinfo {author}
  {\bibfnamefont {M.}~\bibnamefont {Je{\v{z}}ek}}, \bibinfo {author}
  {\bibfnamefont {R.}~\bibnamefont {Filip}},\ and\ \bibinfo {author}
  {\bibfnamefont {J.}~\bibnamefont {Fiur{\'a}{\v{s}}ek}},\ }\bibfield  {title}
  {\bibinfo {title} {Nondestructive detector for exchange symmetry of photonic
  qubits},\ }\href {https://doi.org/10.1038/s41534-018-0087-x} {\bibfield
  {journal} {\bibinfo  {journal} {npj Quantum Inf}\ }\textbf {\bibinfo {volume}
  {4}},\ \bibinfo {pages} {1} (\bibinfo {year} {2018})}\BibitemShut {NoStop}%
\bibitem [{\citenamefont {Yu}\ and\ \citenamefont
  {Ying}(2015)}]{yu_optimal_2015}%
  \BibitemOpen
  \bibfield  {author} {\bibinfo {author} {\bibfnamefont {N.}~\bibnamefont
  {Yu}}\ and\ \bibinfo {author} {\bibfnamefont {M.}~\bibnamefont {Ying}},\
  }\bibfield  {title} {\bibinfo {title} {Optimal simulation of {Deutsch} gates
  and the {Fredkin} gate},\ }\href {https://doi.org/10.1103/PhysRevA.91.032302}
  {\bibfield  {journal} {\bibinfo  {journal} {Phys. Rev. A}\ }\textbf {\bibinfo
  {volume} {91}},\ \bibinfo {pages} {032302} (\bibinfo {year}
  {2015})}\BibitemShut {NoStop}%
\bibitem [{\citenamefont {Pichler}\ \emph {et~al.}(2016)\citenamefont
  {Pichler}, \citenamefont {Zhu}, \citenamefont {Seif}, \citenamefont
  {Zoller},\ and\ \citenamefont {Hafezi}}]{pichler_measurement_2016}%
  \BibitemOpen
  \bibfield  {author} {\bibinfo {author} {\bibfnamefont {H.}~\bibnamefont
  {Pichler}}, \bibinfo {author} {\bibfnamefont {G.}~\bibnamefont {Zhu}},
  \bibinfo {author} {\bibfnamefont {A.}~\bibnamefont {Seif}}, \bibinfo {author}
  {\bibfnamefont {P.}~\bibnamefont {Zoller}},\ and\ \bibinfo {author}
  {\bibfnamefont {M.}~\bibnamefont {Hafezi}},\ }\bibfield  {title} {\bibinfo
  {title} {Measurement protocol for the entanglement spectrum of cold atoms},\
  }\href {https://doi.org/10.1103/PhysRevX.6.041033} {\bibfield  {journal}
  {\bibinfo  {journal} {Phys. Rev. X}\ }\textbf {\bibinfo {volume} {6}},\
  \bibinfo {pages} {041033} (\bibinfo {year} {2016})}\BibitemShut {NoStop}%
\bibitem [{sup()}]{supp}%
  \BibitemOpen
  \href@noop {} {}\bibinfo {note} {See Supplemental Material for the
  Appendixes.}\BibitemShut {Stop}%
\bibitem [{\citenamefont {Foldager}\ and\ \citenamefont
  {Koczor}(2023)}]{foldager_can_2023}%
  \BibitemOpen
  \bibfield  {author} {\bibinfo {author} {\bibfnamefont {J.}~\bibnamefont
  {Foldager}}\ and\ \bibinfo {author} {\bibfnamefont {B.}~\bibnamefont
  {Koczor}},\ }\bibfield  {title} {\bibinfo {title} {Can shallow quantum
  circuits scramble local noise into global white noise?},\ }\href
  {https://doi.org/10.1088/1751-8121/ad0ac7} {\bibfield  {journal} {\bibinfo
  {journal} {J. Phys. A: Math. Theor.}\ }\textbf {\bibinfo {volume} {57}},\
  \bibinfo {pages} {015306} (\bibinfo {year} {2023})}\BibitemShut {NoStop}%
\bibitem [{\citenamefont {Dalzell}\ \emph {et~al.}(2024)\citenamefont
  {Dalzell}, \citenamefont {Hunter-Jones},\ and\ \citenamefont
  {Brand{\~a}o}}]{dalzell_random_2024}%
  \BibitemOpen
  \bibfield  {author} {\bibinfo {author} {\bibfnamefont {A.~M.}\ \bibnamefont
  {Dalzell}}, \bibinfo {author} {\bibfnamefont {N.}~\bibnamefont
  {Hunter-Jones}},\ and\ \bibinfo {author} {\bibfnamefont {F.~G.}\ \bibnamefont
  {Brand{\~a}o}},\ }\bibfield  {title} {\bibinfo {title} {Random quantum
  circuits transform local noise into global white noise},\ }\href
  {https://doi.org/10.1007/s00220-024-04958-z} {\bibfield  {journal} {\bibinfo
  {journal} {Commun. Math. Phys.}\ }\textbf {\bibinfo {volume} {405}},\
  \bibinfo {pages} {78} (\bibinfo {year} {2024})}\BibitemShut {NoStop}%
\bibitem [{\citenamefont {Urbanek}\ \emph {et~al.}(2021)\citenamefont
  {Urbanek}, \citenamefont {Nachman}, \citenamefont {Pascuzzi}, \citenamefont
  {He}, \citenamefont {Bauer},\ and\ \citenamefont
  {de~Jong}}]{urbanek_mitigating_2021}%
  \BibitemOpen
  \bibfield  {author} {\bibinfo {author} {\bibfnamefont {M.}~\bibnamefont
  {Urbanek}}, \bibinfo {author} {\bibfnamefont {B.}~\bibnamefont {Nachman}},
  \bibinfo {author} {\bibfnamefont {V.~R.}\ \bibnamefont {Pascuzzi}}, \bibinfo
  {author} {\bibfnamefont {A.}~\bibnamefont {He}}, \bibinfo {author}
  {\bibfnamefont {C.~W.}\ \bibnamefont {Bauer}},\ and\ \bibinfo {author}
  {\bibfnamefont {W.~A.}\ \bibnamefont {de~Jong}},\ }\bibfield  {title}
  {\bibinfo {title} {Mitigating depolarizing noise on quantum computers with
  noise-estimation circuits},\ }\href
  {https://doi.org/10.1103/PhysRevLett.127.270502} {\bibfield  {journal}
  {\bibinfo  {journal} {Phys. Rev. Lett.}\ }\textbf {\bibinfo {volume} {127}},\
  \bibinfo {pages} {270502} (\bibinfo {year} {2021})}\BibitemShut {NoStop}%
\bibitem [{\citenamefont {Mi}\ \emph {et~al.}(2021)\citenamefont {Mi},
  \citenamefont {Roushan}, \citenamefont {Quintana}, \citenamefont {Mandrà},
  \citenamefont {Marshall}, \citenamefont {Neill}, \citenamefont {Arute},
  \citenamefont {Arya}, \citenamefont {Atalaya}, \citenamefont {Babbush},
  \citenamefont {Bardin}, \citenamefont {Barends}, \citenamefont {Basso},
  \citenamefont {Bengtsson}, \citenamefont {Boixo}, \citenamefont {Bourassa},
  \citenamefont {Broughton}, \citenamefont {Buckley}, \citenamefont {Buell},
  \citenamefont {Burkett}, \citenamefont {Bushnell}, \citenamefont {Chen},
  \citenamefont {Chiaro}, \citenamefont {Collins}, \citenamefont {Courtney},
  \citenamefont {Demura}, \citenamefont {Derk}, \citenamefont {Dunsworth},
  \citenamefont {Eppens}, \citenamefont {Erickson}, \citenamefont {Farhi},
  \citenamefont {Fowler}, \citenamefont {Foxen}, \citenamefont {Gidney},
  \citenamefont {Giustina}, \citenamefont {Gross}, \citenamefont {Harrigan},
  \citenamefont {Harrington}, \citenamefont {Hilton}, \citenamefont {Ho},
  \citenamefont {Hong}, \citenamefont {Huang}, \citenamefont {Huggins},
  \citenamefont {Ioffe}, \citenamefont {Isakov}, \citenamefont {Jeffrey},
  \citenamefont {Jiang}, \citenamefont {Jones}, \citenamefont {Kafri},
  \citenamefont {Kelly}, \citenamefont {Kim}, \citenamefont {Kitaev},
  \citenamefont {Klimov}, \citenamefont {Korotkov}, \citenamefont {Kostritsa},
  \citenamefont {Landhuis}, \citenamefont {Laptev}, \citenamefont {Lucero},
  \citenamefont {Martin}, \citenamefont {McClean}, \citenamefont {McCourt},
  \citenamefont {McEwen}, \citenamefont {Megrant}, \citenamefont {Miao},
  \citenamefont {Mohseni}, \citenamefont {Montazeri}, \citenamefont
  {Mruczkiewicz}, \citenamefont {Mutus}, \citenamefont {Naaman}, \citenamefont
  {Neeley}, \citenamefont {Newman}, \citenamefont {Niu}, \citenamefont
  {O’Brien}, \citenamefont {Opremcak}, \citenamefont {Ostby}, \citenamefont
  {Pato}, \citenamefont {Petukhov}, \citenamefont {Redd}, \citenamefont
  {Rubin}, \citenamefont {Sank}, \citenamefont {Satzinger}, \citenamefont
  {Shvarts}, \citenamefont {Strain}, \citenamefont {Szalay}, \citenamefont
  {Trevithick}, \citenamefont {Villalonga}, \citenamefont {White},
  \citenamefont {Yao}, \citenamefont {Yeh}, \citenamefont {Zalcman},
  \citenamefont {Neven}, \citenamefont {Aleiner}, \citenamefont {Kechedzhi},
  \citenamefont {Smelyanskiy},\ and\ \citenamefont
  {Chen}}]{mi_information_2021}%
  \BibitemOpen
  \bibfield  {author} {\bibinfo {author} {\bibfnamefont {X.}~\bibnamefont
  {Mi}}, \bibinfo {author} {\bibfnamefont {P.}~\bibnamefont {Roushan}},
  \bibinfo {author} {\bibfnamefont {C.}~\bibnamefont {Quintana}}, \bibinfo
  {author} {\bibfnamefont {S.}~\bibnamefont {Mandrà}}, \bibinfo {author}
  {\bibfnamefont {J.}~\bibnamefont {Marshall}}, \bibinfo {author}
  {\bibfnamefont {C.}~\bibnamefont {Neill}}, \bibinfo {author} {\bibfnamefont
  {F.}~\bibnamefont {Arute}}, \bibinfo {author} {\bibfnamefont
  {K.}~\bibnamefont {Arya}}, \bibinfo {author} {\bibfnamefont {J.}~\bibnamefont
  {Atalaya}}, \bibinfo {author} {\bibfnamefont {R.}~\bibnamefont {Babbush}},
  \bibinfo {author} {\bibfnamefont {J.~C.}\ \bibnamefont {Bardin}}, \bibinfo
  {author} {\bibfnamefont {R.}~\bibnamefont {Barends}}, \bibinfo {author}
  {\bibfnamefont {J.}~\bibnamefont {Basso}}, \bibinfo {author} {\bibfnamefont
  {A.}~\bibnamefont {Bengtsson}}, \bibinfo {author} {\bibfnamefont
  {S.}~\bibnamefont {Boixo}}, \bibinfo {author} {\bibfnamefont
  {A.}~\bibnamefont {Bourassa}}, \bibinfo {author} {\bibfnamefont
  {M.}~\bibnamefont {Broughton}}, \bibinfo {author} {\bibfnamefont {B.~B.}\
  \bibnamefont {Buckley}}, \bibinfo {author} {\bibfnamefont {D.~A.}\
  \bibnamefont {Buell}}, \bibinfo {author} {\bibfnamefont {B.}~\bibnamefont
  {Burkett}}, \bibinfo {author} {\bibfnamefont {N.}~\bibnamefont {Bushnell}},
  \bibinfo {author} {\bibfnamefont {Z.}~\bibnamefont {Chen}}, \bibinfo {author}
  {\bibfnamefont {B.}~\bibnamefont {Chiaro}}, \bibinfo {author} {\bibfnamefont
  {R.}~\bibnamefont {Collins}}, \bibinfo {author} {\bibfnamefont
  {W.}~\bibnamefont {Courtney}}, \bibinfo {author} {\bibfnamefont
  {S.}~\bibnamefont {Demura}}, \bibinfo {author} {\bibfnamefont {A.~R.}\
  \bibnamefont {Derk}}, \bibinfo {author} {\bibfnamefont {A.}~\bibnamefont
  {Dunsworth}}, \bibinfo {author} {\bibfnamefont {D.}~\bibnamefont {Eppens}},
  \bibinfo {author} {\bibfnamefont {C.}~\bibnamefont {Erickson}}, \bibinfo
  {author} {\bibfnamefont {E.}~\bibnamefont {Farhi}}, \bibinfo {author}
  {\bibfnamefont {A.~G.}\ \bibnamefont {Fowler}}, \bibinfo {author}
  {\bibfnamefont {B.}~\bibnamefont {Foxen}}, \bibinfo {author} {\bibfnamefont
  {C.}~\bibnamefont {Gidney}}, \bibinfo {author} {\bibfnamefont
  {M.}~\bibnamefont {Giustina}}, \bibinfo {author} {\bibfnamefont {J.~A.}\
  \bibnamefont {Gross}}, \bibinfo {author} {\bibfnamefont {M.~P.}\ \bibnamefont
  {Harrigan}}, \bibinfo {author} {\bibfnamefont {S.~D.}\ \bibnamefont
  {Harrington}}, \bibinfo {author} {\bibfnamefont {J.}~\bibnamefont {Hilton}},
  \bibinfo {author} {\bibfnamefont {A.}~\bibnamefont {Ho}}, \bibinfo {author}
  {\bibfnamefont {S.}~\bibnamefont {Hong}}, \bibinfo {author} {\bibfnamefont
  {T.}~\bibnamefont {Huang}}, \bibinfo {author} {\bibfnamefont {W.~J.}\
  \bibnamefont {Huggins}}, \bibinfo {author} {\bibfnamefont {L.~B.}\
  \bibnamefont {Ioffe}}, \bibinfo {author} {\bibfnamefont {S.~V.}\ \bibnamefont
  {Isakov}}, \bibinfo {author} {\bibfnamefont {E.}~\bibnamefont {Jeffrey}},
  \bibinfo {author} {\bibfnamefont {Z.}~\bibnamefont {Jiang}}, \bibinfo
  {author} {\bibfnamefont {C.}~\bibnamefont {Jones}}, \bibinfo {author}
  {\bibfnamefont {D.}~\bibnamefont {Kafri}}, \bibinfo {author} {\bibfnamefont
  {J.}~\bibnamefont {Kelly}}, \bibinfo {author} {\bibfnamefont
  {S.}~\bibnamefont {Kim}}, \bibinfo {author} {\bibfnamefont {A.}~\bibnamefont
  {Kitaev}}, \bibinfo {author} {\bibfnamefont {P.~V.}\ \bibnamefont {Klimov}},
  \bibinfo {author} {\bibfnamefont {A.~N.}\ \bibnamefont {Korotkov}}, \bibinfo
  {author} {\bibfnamefont {F.}~\bibnamefont {Kostritsa}}, \bibinfo {author}
  {\bibfnamefont {D.}~\bibnamefont {Landhuis}}, \bibinfo {author}
  {\bibfnamefont {P.}~\bibnamefont {Laptev}}, \bibinfo {author} {\bibfnamefont
  {E.}~\bibnamefont {Lucero}}, \bibinfo {author} {\bibfnamefont
  {O.}~\bibnamefont {Martin}}, \bibinfo {author} {\bibfnamefont {J.~R.}\
  \bibnamefont {McClean}}, \bibinfo {author} {\bibfnamefont {T.}~\bibnamefont
  {McCourt}}, \bibinfo {author} {\bibfnamefont {M.}~\bibnamefont {McEwen}},
  \bibinfo {author} {\bibfnamefont {A.}~\bibnamefont {Megrant}}, \bibinfo
  {author} {\bibfnamefont {K.~C.}\ \bibnamefont {Miao}}, \bibinfo {author}
  {\bibfnamefont {M.}~\bibnamefont {Mohseni}}, \bibinfo {author} {\bibfnamefont
  {S.}~\bibnamefont {Montazeri}}, \bibinfo {author} {\bibfnamefont
  {W.}~\bibnamefont {Mruczkiewicz}}, \bibinfo {author} {\bibfnamefont
  {J.}~\bibnamefont {Mutus}}, \bibinfo {author} {\bibfnamefont
  {O.}~\bibnamefont {Naaman}}, \bibinfo {author} {\bibfnamefont
  {M.}~\bibnamefont {Neeley}}, \bibinfo {author} {\bibfnamefont
  {M.}~\bibnamefont {Newman}}, \bibinfo {author} {\bibfnamefont {M.~Y.}\
  \bibnamefont {Niu}}, \bibinfo {author} {\bibfnamefont {T.~E.}\ \bibnamefont
  {O’Brien}}, \bibinfo {author} {\bibfnamefont {A.}~\bibnamefont {Opremcak}},
  \bibinfo {author} {\bibfnamefont {E.}~\bibnamefont {Ostby}}, \bibinfo
  {author} {\bibfnamefont {B.}~\bibnamefont {Pato}}, \bibinfo {author}
  {\bibfnamefont {A.}~\bibnamefont {Petukhov}}, \bibinfo {author}
  {\bibfnamefont {N.}~\bibnamefont {Redd}}, \bibinfo {author} {\bibfnamefont
  {N.~C.}\ \bibnamefont {Rubin}}, \bibinfo {author} {\bibfnamefont
  {D.}~\bibnamefont {Sank}}, \bibinfo {author} {\bibfnamefont {K.~J.}\
  \bibnamefont {Satzinger}}, \bibinfo {author} {\bibfnamefont {V.}~\bibnamefont
  {Shvarts}}, \bibinfo {author} {\bibfnamefont {D.}~\bibnamefont {Strain}},
  \bibinfo {author} {\bibfnamefont {M.}~\bibnamefont {Szalay}}, \bibinfo
  {author} {\bibfnamefont {M.~D.}\ \bibnamefont {Trevithick}}, \bibinfo
  {author} {\bibfnamefont {B.}~\bibnamefont {Villalonga}}, \bibinfo {author}
  {\bibfnamefont {T.}~\bibnamefont {White}}, \bibinfo {author} {\bibfnamefont
  {Z.~J.}\ \bibnamefont {Yao}}, \bibinfo {author} {\bibfnamefont
  {P.}~\bibnamefont {Yeh}}, \bibinfo {author} {\bibfnamefont {A.}~\bibnamefont
  {Zalcman}}, \bibinfo {author} {\bibfnamefont {H.}~\bibnamefont {Neven}},
  \bibinfo {author} {\bibfnamefont {I.}~\bibnamefont {Aleiner}}, \bibinfo
  {author} {\bibfnamefont {K.}~\bibnamefont {Kechedzhi}}, \bibinfo {author}
  {\bibfnamefont {V.}~\bibnamefont {Smelyanskiy}},\ and\ \bibinfo {author}
  {\bibfnamefont {Y.}~\bibnamefont {Chen}},\ }\bibfield  {title} {\bibinfo
  {title} {Information scrambling in quantum circuits},\ }\href
  {https://doi.org/10.1126/science.abg5029} {\bibfield  {journal} {\bibinfo
  {journal} {Science}\ }\textbf {\bibinfo {volume} {374}},\ \bibinfo {pages}
  {1479} (\bibinfo {year} {2021})}\BibitemShut {NoStop}%
\bibitem [{\citenamefont {Bennett}\ \emph {et~al.}(1996)\citenamefont
  {Bennett}, \citenamefont {Brassard}, \citenamefont {Popescu}, \citenamefont
  {Schumacher}, \citenamefont {Smolin},\ and\ \citenamefont
  {Wootters}}]{bennett_purification_1996}%
  \BibitemOpen
  \bibfield  {author} {\bibinfo {author} {\bibfnamefont {C.~H.}\ \bibnamefont
  {Bennett}}, \bibinfo {author} {\bibfnamefont {G.}~\bibnamefont {Brassard}},
  \bibinfo {author} {\bibfnamefont {S.}~\bibnamefont {Popescu}}, \bibinfo
  {author} {\bibfnamefont {B.}~\bibnamefont {Schumacher}}, \bibinfo {author}
  {\bibfnamefont {J.~A.}\ \bibnamefont {Smolin}},\ and\ \bibinfo {author}
  {\bibfnamefont {W.~K.}\ \bibnamefont {Wootters}},\ }\bibfield  {title}
  {\bibinfo {title} {Purification of noisy entanglement and faithful
  teleportation via noisy channels},\ }\href
  {https://doi.org/10.1103/PhysRevLett.76.722} {\bibfield  {journal} {\bibinfo
  {journal} {Phys. Rev. Lett.}\ }\textbf {\bibinfo {volume} {76}},\ \bibinfo
  {pages} {722} (\bibinfo {year} {1996})}\BibitemShut {NoStop}%
\bibitem [{\citenamefont {Horodecki}\ and\ \citenamefont
  {Horodecki}(1999)}]{horodecki_reduction_1999}%
  \BibitemOpen
  \bibfield  {author} {\bibinfo {author} {\bibfnamefont {M.}~\bibnamefont
  {Horodecki}}\ and\ \bibinfo {author} {\bibfnamefont {P.}~\bibnamefont
  {Horodecki}},\ }\bibfield  {title} {\bibinfo {title} {Reduction criterion of
  separability and limits for a class of distillation protocols},\ }\href
  {https://doi.org/10.1103/PhysRevA.59.4206} {\bibfield  {journal} {\bibinfo
  {journal} {Phys. Rev. A}\ }\textbf {\bibinfo {volume} {59}},\ \bibinfo
  {pages} {4206} (\bibinfo {year} {1999})}\BibitemShut {NoStop}%
\bibitem [{\citenamefont {Liu}\ \emph {et~al.}(2023)\citenamefont {Liu},
  \citenamefont {Li}, \citenamefont {Shang},\ and\ \citenamefont
  {Zhang}}]{Liu_homo_2023}%
  \BibitemOpen
  \bibfield  {author} {\bibinfo {author} {\bibfnamefont {Y.-C.}\ \bibnamefont
  {Liu}}, \bibinfo {author} {\bibfnamefont {Y.}~\bibnamefont {Li}}, \bibinfo
  {author} {\bibfnamefont {J.}~\bibnamefont {Shang}},\ and\ \bibinfo {author}
  {\bibfnamefont {X.}~\bibnamefont {Zhang}},\ }\bibfield  {title} {\bibinfo
  {title} {Efficient verification of arbitrary entangled states with
  homogeneous local measurements},\ }\href
  {https://doi.org/https://doi.org/10.1002/qute.202300083} {\bibfield
  {journal} {\bibinfo  {journal} {Adv. Quantum Technol.}\ }\textbf {\bibinfo
  {volume} {6}},\ \bibinfo {pages} {2300083} (\bibinfo {year}
  {2023})}\BibitemShut {NoStop}%
\end{thebibliography}
\end{document}